\documentclass[11pt]{article}
\usepackage{amsmath,amssymb,amsthm,amsfonts,hhline,color,multicol,multirow,hyperref,booktabs}
\usepackage{etoolbox}
\usepackage{authblk}

\def\cN{N}

\def\F{{\mathbb F}}
\def\Z{{\mathbb Z}}

\newcommand{\voa}{vertex operator algebra}

\newcommand{\voas}{vertex operator algebras}
\newcommand{\cft}{conformal field theory}
\newcommand{\cfts}{conformal field theories}
\newcommand{\scft}{superconformal field theory}

\newcommand{\ZZ}{\mathbb{Z}}
\newcommand{\CC}{\mathbb{C}}
\newcommand{\vv}{\mathbf{v}}
\newcommand{\lab}{\langle}
\newcommand{\rab}{\rangle}
\newcommand{\Id}{I}
\newcommand{\tr}{\operatorname{tr}}
\newcommand{\ch}{\operatorname{ch}}
\newcommand{\chh}{\widehat{\operatorname{ch}}}
\newcommand{\wt}{\operatorname{wt}}
\newcommand{\vosa}{vertex operator superalgebra}
\newcommand{\vosas}{vertex operator superalgebras}
\newcommand{\mfg}{\mathfrak{g}}
\newcommand{\SL}{\operatorname{\textsl{SL}}}
\newcommand{\EG}{\operatorname{\it E}}
\newcommand{\scfts}{superconformal field theories}
\newcommand{\NSNS}{\text{NS-NS}}
\newcommand{\RR}{\text{R-R}}
\usepackage{hyperref}
\usepackage{setspace}
\def\cH{{\mathcal H}}
\def\cD{{\mathcal D}}
\def\cG{{\mathcal G}}
\newcommand{\even}{\mathrm{even}}
\newcommand{\odd}{\mathrm{odd}}
\newcommand{\tw}{\mathrm{tw}}
\newcommand{\xmod}{{\rm \;mod\;}}
\newcommand{\comment}[1]{}
\newcommand{\ldp}{((}
\newcommand{\rdp}{))}
\newcommand{\Co}{\textsl{Co}}

\begin{document}

\title{\textsc{Self-Dual Vertex Operator Superalgebras\\ and Superconformal Field Theory}}

\renewcommand{\thefootnote}{\arabic{footnote}}

\author[1]{Thomas Creutzig\thanks{ 
Email: creutzig@ualberta.ca}}
\author[2]{John F. R. Duncan\thanks{
Email: john.duncan@emory.edu}}
\author[1]{Wolfgang Riedler\thanks{ 
Email: riedler@ualberta.ca}} 

\affil[1]{Department of Mathematical and Statistical Sciences, University of Alberta, Edmonton, Alberta  T6G 2G1, Canada.}
\affil[2]{Department of Mathematics and Computer Science, Emory University, Atlanta, GA 30322, U.S.A.}

\date{}

\setstretch{1.15}

\theoremstyle{plain}
\newtheorem*{introthm}{Theorem}
\newtheorem{thm}{Theorem}[section]
\newtheorem{prop}[thm]{Proposition}
\newtheorem{lem}[thm]{Lemma}
\newtheorem{cor}[thm]{Corollary}
\newtheorem{conj}[thm]{Conjecture}

\theoremstyle{definition}
\newtheorem{defn}[thm]{Definition}
\newtheorem{rem}[thm]{Remark}

\numberwithin{equation}{section}

\renewcommand{\thefootnote}{\fnsymbol{footnote}} 
\footnotetext{\emph{MSC2010:} 17B69, 17B81, 20C34.}     

\maketitle

\begin{abstract}
Recent work has related the equivariant elliptic genera of sigma models with K3 surface target to a vertex operator superalgebra that realizes moonshine for Conway's group. Motivated by this we consider conditions under which a self-dual vertex operator superalgebra may be identified with the bulk Hilbert space of a superconformal field theory. After presenting a classification result for self-dual vertex operator superalgebras with central charge up to 12 we describe several examples of close relationships with bulk superconformal field theories, including those arising from sigma models for tori and K3 surfaces. 
\end{abstract}

\clearpage

\tableofcontents

\section{Introduction}\label{sec:intro}

The 
elliptic genus of a complex K3 surface $X$ is a weak Jacobi form of weight zero and index one. It may be realized in the following three ways: via the chiral de Rham complex of $X$ \cite{B-ZHS,Hel,B,BL}, as the $S^1$-equivariant $\chi_y$-genus of the loop space of $X$ 
\cite{Hi, Ho, Kr},
and as a trace 
on the Ramond-Ramond sector of 
a sigma model on $X$ \cite{EOTY,DY,W-LG}.  
The small $\cN=4$ superconformal algebra at central charge $c=6$ appears in each of these three pictures. Firstly, it is the algebra of 
global sections of the chiral de Rham complex \cite{So1, So2}. Secondly, the $\chi_y$-genus can be viewed as a virtual module for a certain \vosa{} that contains this Lie superalgebra 
\cite{CH, Ta}. Finally, it appears as a supersymmetry algebra of the string theory sigma model \cite{ET}. (We refer to \cite{W-snap} for a recent detailed review of these topics.)

The K3 elliptic genus and the $\cN=4$ superconformal algebra were absorbed into the orbit of moonshine when Eguchi--Ooguri--Tachikawa \cite{EOT} suggested a relationship between the largest Mathieu group, $M_{24}$, and character contributions of the $\cN=4$ superconformal algebra to the K3 elliptic genus. 
This ignited a resurgence of interest in connections between string theory, modular forms and finite groups. Umbral moonshine \cite{UM,MUM,umrec}, Thompson moonshine \cite{HR,GM16} and the recently announced O'Nan moonshine \cite{DMO} all belong to the quickly developing legacy of this {\em Mathieu moonshine} observation, although the connections to string theory are so far more obscure in the latter two cases. We refer to \cite{mnstmlts} for a fuller review, more references, and for comparison to the original monstrous moonshine \cite{MR554399,Tho_FinGpsModFns,Tho_NmrlgyMonsEllModFn} that appeared in the 1970s.

By now there are  indications that Mathieu and monstrous moonshine are interrelated. 
An instance of this, and a primary motivation for the present work is \cite{DM16}, wherein
the K3 elliptic genus apparently makes a fourth appearance: as a trace function on the 
moonshine module \cite{Dun_VACo,DM15} for Conway's group, $\Co_0$ \cite{MR0237634,MR0248216}. 
On the one hand, this Conway moonshine module---a vertex operator superalgebra with $\cN=1$ structure---is a direct supersymmetric analogue of the monstrous moonshine module of Frenkel--Lepowsky--Meurman \cite{FLMPNAS,FLMBerk,FLM}. 
It manifests a genus zero property for 
$\Co_0$
\cite{DM15}, just as the monstrous moonshine module does for the monster \cite{borcherds_monstrous}.
On the other hand, $M_{24}$ is a subgroup of $\Co_0$, and the Conway moonshine construction of the K3 elliptic genus may be twined by (most) elements of $M_{24}$. In many, but not all instances the resulting trace functions coincide with those that arise in Mathieu moonshine.

This tells us that 
the Conway moonshine module comes close to providing a vertex algebraic realization of the as yet elusive 
Mathieu moonshine module, 
whose structure as a representation of $M_{24}$ was conjecturally determined in
\cite{Cheng,Gaberdiel2010,Gaberdiel2010a,Eguchi2010a}, and confirmed 
in \cite{Ga}. 
The Conway moonshine module has been used 
to realize analogues of the Mathieu moonshine module for other sporadic simple groups in \cite{CDDHK,CHKW}.
See \cite{TW12,TW13,TW15,GKV,TW17} for the development of a promising geometric approach to the problem.

The Conway moonshine module is also connected to string theory on K3 surfaces.
As is explained in \cite{DM16}, 
the Conway moonshine 
construction of the K3 elliptic genus may also be twined, in an explicitly computable way, by
any automorphism of a K3 sigma model that preserves its supersymmetry. Such automorphisms are classified in \cite{GHV}. It appears that 
the construction of \cite{DM16} 
is (except for a small number of possible exceptions) in 
agreement \cite{CHVZ} with the twined K3 elliptic genera that one expects \cite{CHVZ} to arise from string theory. 
One reason this is surprising is that K3 sigma models, and in particular their automorphisms, are difficult to construct in general (cf. e.g. \cite{NW}). The Conway moonshine module seems to serve as a shortcut, to certain computations which might otherwise require the explicit construction of sigma models.

In view of these connections 
it is natural to ask how the Conway moonshine realization of the K3 elliptic genus is related to the three we began with above. 
Katz--Klemm--Vafa \cite{KKV} conjectured a method for computing the Gromov--Witten invariants of a K3 surface in terms of the
$\chi_y$-genera of its symmetric powers, and the generating function of these $\chi_y$-genera can be realized in terms of a lift of the K3 elliptic genus. So 
the second mentioned realization of the K3 elliptic genus, as a generalization of Hirzebruch's $\chi_y$-genus, suggests a connection between Conway moonshine and enumerative geometry. This perspective is developed in \cite{CDHK}, where equivariant counterparts to the conjecture of Katz--Klemm--Vafa
are formulated, which explicitly describe equivariant versions of Gromov--Witten invariants of K3 surfaces. The conjecture of Katz--Klemm--Vafa  was proved recently by Pandharipande--Thomas \cite{PT}. Katz--Klemm--Pandharipande \cite{KKP} have extended the the conjecture of Katz--Klemm--Vafa to refined Gopakumar--Vafa invariants. Conjectural descriptions of equivariant refined Gopakumar--Vafa invariants of K3 surfaces are also formulated in \cite{CDHK}.

In this work we develop the relationship between Conway moonshine and the third mentioned realization, in terms of K3 sigma models. We do this by formalizing a new relationship between vertex algebra and conformal field theory, and 
by 
realizing the Conway moonshine module, and other vertex operator superalgebras, in examples.

Traditionally, vertex operator algebras satisfying suitable conditions are considered to define 
``chiral halves'' of conformal field theories. More specifically, the bulk Hilbert space of a conformal field theory may be regarded as (a completion of) a suitable sum of modules for a tensor product of vertex operator algebras (cf. \cite{MR1160677,Gab,W-snap}).
The alternative viewpoint we pursue here develops from the observation 
that a bulk Hilbert space of a conformal field theory, taken as a whole, resembles a self-dual vertex operator algebra. We explain this observation more fully in \S\ref{sec:bcft:pbcft}. It motivates our 
\begin{quote}
\noindent{\bf Main Question:} {\em Can a self-dual \voa{} be identified with a bulk \cft{} in some sense?}
\end{quote}
We answer this question positively by formulating 
the notion of {\em potential (bulk) \cft{}} (cf. Definition \ref{def:pbcft}) and by identifying self-dual vertex operator algebras as examples (cf. Propositions \ref{prop:VD2ncft} and \ref{prop:diagcftVL}). In fact, we formulate supersymmetric counterparts to potential conformal field theories as well (cf. Definitions \ref{def:pbscft} and \ref{def:qpb22scft}), and find more examples amongst self-dual vertex operator superalgebras (cf. Theorems \ref{thm:D4n+F4n}, \ref{thm:D12+LL} and \ref{thm:WKK}). To support the analysis we also present a classification result (Theorem \ref{thm:selfdual-class}) for self-dual vertex operator superalgebras with central charge up to $12$.

Equipped with the notion of potential bulk \scft{} we relate the Conway moonshine module to four \scfts{} in \S\ref{sec:egs}. One of these is the \scft{} underlying the tetrahedral K3 sigma model (cf. \S\ref{sec:egs:typeak3}), which was analyzed in detail by Gaberdiel--Taormina--Volpato--Wendland \cite{GTVW} (see also \cite{TW17}). 
Another is the Gepner model $(1)^6$ (cf. \S\ref{sec:egs:gepK3}), and it is clear that there are further interesting examples waiting to be considered, that may shed more light on the role of Conway moonshine in K3 string theory.

An implication of our analysis is that there should be self-dual vertex operator superalgebras besides the Conway moonshine module that have analogous relationships to other string theory compacitifcations. In \S\ref{sec:egs:typeDscft} we identify a self-dual vertex operator superalgebra---the $\cN=1$ \vosa{} naturally attached to the $E_8$ lattice---which realizes the bulk \scft{} underlying a sigma model with a $4$-torus as target (cf. Theorem \ref{thm:D4n+F4n}). Volpato \cite{V} has shown that the supersymmetry preserving automorphism groups of $4$-torus sigma models are subgroups of the Weyl group of the $E_8$ lattice. In light of these results it seems likely 
that 
the $E_8$ \vosa{} that appears in \S\ref{sec:egs:typeak3} can serve as a counterpart to the Conway moonshine module for nonlinear sigma models on $4$-dimensional tori.

It is interesting to compare the approach presented here to recent work \cite{TW17} of Taormina--Wendland. In loc. cit. the relationship between superconformal field theory and vertex operator superalgebra is also reconsidered, but the starting point is a fully fledged superconformal field theory. A notion of {\em reflection} is introduced which, in special circumstances, produces a vertex operator superalgebra. The \scft{} underlying the tetrahedral K3 sigma model is considered in detail, and it is shown that the Conway moonshine module arises when reflection is applied in this case. In this way Taormina--Wendland independently obtain 
results equivalent to those we present in \S\ref{sec:egs:typeak3}. Our notion of potential \scft{} serves to answer the question of what reflection produces from a \scft{} in general, except that a reflected \scft{} comes equipped with extra structure, on account of the richness of the \scft{} axioms. For example, the reflected tetrahedral K3 theory recovers the \vosa{} structure on the Conway moonshine module, but also furnishes an intertwining operator algebra structure on the direct sum of itself with its unique irreducible canonically twisted module (cf. \S4 of \cite{TW17}). Moving forward, we can expect that Taormina--Wendland reflection will play a key role in further elucidating the relationships between \scfts{}, potential \scfts{} and \vosas{}.

We now describe the structure of the article. We present background material in \S\ref{sec:back}. We explain our conventions on \vosas{} in \S\ref{sec:back:vosas}, and we review some modularity results for \vosas{} in \S\ref{sec:back:mod}. We recall the small $\cN=4$ superconformal algebra in \S\ref{sec:back:sca}, and describe an explicit construction at $c=6$ in \S\ref{sec:back:ff}. In \S\ref{sec:selfdual} we establish our classification result for self-dual \vosas{} with central charge at most $12$. Then in \S\ref{sec:bscft} we discuss the new relationship between vertex algebra and \cft{} that motivates this work, and explain our approach to answering the Main Question. We begin with \cft{} in \S\ref{sec:bcft:pbcft}, discuss \scft{} in \S\ref{sec:bscft:pbscft}, and consider \scfts{} with superconformal structure in \S\ref{sec:bscft:supstruc}. We present examples of bulk \scft{} interpretations of self-dual vertex operator superalgebras in \S\ref{sec:egs}. We begin with the diagonal \cfts{} associated to type $D$ lattice \voas{} in \S\ref{sec:egs:typeDcft}, and then discuss super analogues of these in \S\ref{sec:egs:typeDscft}. We discuss the \scft{} underlying the tetrahedral K3 sigma model in \S\ref{sec:egs:typeak3}, and discuss the relationship between the Conway moonshine module and the Gepner model $(1)^6$ in \S\ref{sec:egs:gepK3}.

\section{Background}\label{sec:back}

\subsection{Vertex Superalgebra}\label{sec:back:vosas}

We assume some familiarity with the basics of vertex (operator) superalgebra theory. Good references for this include \cite{FB-Z,FLM,Kac,LL}.

We adopt the 
convention, common in physical settings, of writing $(-1)^F$ for the canonical involution on a superspace $W=W^\even\oplus W^\odd$, so that $(-1)^F|_{W^\even}=\Id$ and $(-1)^F|_{W^\odd}=-\Id$.
We write $Y(a,z)=\sum a_{(k)}z^{-k-1}$ for the vertex operator attached to an element $a$ in a vertex superalgebra 
$W$.
A vertex superalgebra $W$ is called {\em $C_2$-cofinite} if $W/C_2(W)$ is finite-dimensional, where
$C_2(W):=\{a_{(-2)}b\mid a,b\in W\}$. Following \cite{MR1650625} we say that a vertex operator superalgebra $W$ is of {\em CFT type} if the $L_0$-grading $W=\bigoplus_{n\in\frac{1}{2}\ZZ}W_n$ is bounded below by $0$, and if $W_0$ is spanned by the vacuum vector. We assume that $W^\even=\bigoplus_{n\in \ZZ}W_n$ and $W^\odd=\bigoplus_{n\in \ZZ+\frac12}W_n$.
A \vosa{} that is $C_2$-cofinite and of CFT type is {\em nice (sch\"on)} in the sense of \cite{Hoehn2007}.

Say that a vertex operator algebra is {\em rational} if all of its admissible modules are completely reducible. 
We refer to \cite{DLM1998} for the definition of admissible module. 
It is proven in loc. cit. that a rational vertex operator algebra has finitely many irreducible admissible modules up to equivalence.
We say that a vertex operator superalgebra is {rational} if its even sub vertex operator algebra is rational. 
We will apply results from \cite{DZ} in what follows, so we should note that our notion of rationality for a vertex operator superalgebra is stronger than that which appears there. A vertex operator superalgebra that is rational in our sense is both rational and $(-1)^F$-rational in the sense of loc. cit. The equivalence of the two notions of rationality  is proven in \cite{HA} under an assumption on fusion products of canonically twisted modules.

We say that a \vosa{} $W$ is {\em self-dual} if $W$ is rational (in our sense), irreducible as a $W$-module, and if $W$ is the only irreducible admissible $W$-module up to isomorphism. Note that the term self-dual is sometimes used differently elsewhere in the literature, to refer to the situation in which $W$ is isomorphic to its contragredient as a $W$-module. 

According to Theorem 8.7 of \cite{DZ} a self-dual $C_2$-cofinite \vosa{} $W$ has a unique (up to isomorphism) irreducible $(-1)^F$-stable canonically twisted module. We denote it $W_\tw$. The $(-1)^F$-stable condition on a canonically twisted module $M$ for $W$ is equivalent to the requirement of a superspace structure $M=M^\even\oplus M^\odd$ that is compatible with the superspace structure  on $W$, so that elements of $W^\even$ and $W^\odd$ induce even and odd transformations of $M$, respectively. Modules that are not $(-1)^F$-stable will not arise in this work so we henceforth assume the existence of a compatible superspace structure to be a part of the definition of untwisted or canonically twisted module for a \vosa{}. However, we will not require morphisms of modules to preserve a particular superspace structure. So for example, if $W$ is a \vosa{} and $\Pi$ is the parity change functor on superspaces then $W$ and $\Pi W$ are not isomorphic as superspaces, but we do regard them as isomorphic $W$-modules.

Write $V_L$ for the vertex superalgebra attached to an integral lattice $L$, which is naturally a \vosa{} if $L$ is positive definite. Write $F(n)$ for the \vosa{} of $n$ free fermions. According to the boson-fermion correspondence \cite{Frenkel81,DM94} the \vosa{} attached to $\ZZ^n$ is isomorphic to $F(2n)$.
So the even sub vertex operator algebra $F(2n)^\even<F(2n)$ is isomorphic to the lattice vertex operator algebra attached to the {\em type D} lattice
\begin{gather}
D_n:=\left\{ (x_1, \dots, x_n) \in \Z^n\mid x_1+\dots+x_n =0\xmod{2} \right\}.
\end{gather}
The discriminant group of $D_n$ is $D_n^*/D_n\simeq \ZZ/2\ZZ\times \ZZ/2\ZZ$, and we label coset representatives as follows.
\begin{gather}
\begin{split}\label{eqn:[i]vector}
[0]&:= (0, \dots, 0,0), \quad
[1]:= \frac{1}{2}(1, \dots, 1,1), \\
[2]&:= (0, \dots, 0, 1), \quad
[3]:= \frac{1}{2}(1, \dots, 1, -1).
\end{split}
\end{gather}
Set $D_n^+:=D_n\cup D_n+[1]$. Then $D_n^+$ is a self-dual integral lattice---the rank $n$ {\em spin lattice}---whenever $n=0\xmod 4$. It is even if $n=0\xmod 8$. We have $D_4^+\cong \ZZ^4$ and $D_8^+\cong E_8$, and $D_{12}^+$ is the unique self-dual integral lattice of rank $12$ such that $\lambda\cdot\lambda\leq 1$ implies $\lambda=0$. The lattice \vosas{} attached to $D_n$ and $D_{4n}^+$ will play a prominent role later on.

Set $A_1=\sqrt{2}\ZZ$. 
We will make use of the fact that $D_{2n}$ admits $A_1^{2n}$ as a sub lattice. 
Explicitly, denoting $e_1:=(1,0,\dots 0)$, $e_2:=(0,1,\dots,0)$, \&c., we may take the first copy of $A_1$ to be generated by $e_1+e_{n+1}$, the second copy to be generated by $e_1-e_{n+1}$, the third copy to be generated by $e_2+e_{n+2}$, \&c. In the case that $n=2$ this embedding is actually an isomorphism, $D_2\cong A_1\oplus A_1$. More generally, $D_{2n}/A_1^{2n}$ embeds in the discriminant group of $A_1^{2n}$, which is naturally isomorphic to $(\ZZ/2\ZZ)^{2n}\cong \F_2^{2n}$. As such, it is natural to use binary codewords of length $2n$ to label cosets of $A_1^{2n}$ in its dual. Given such a codeword $C\in \F_2^{2n}$, define $\wt(C)$---the {\em weight} of $C$---to be the number of non-zero entries of $C$. Define a binary code $\cD_{2n}<\F_2^{2n}$ by setting
\begin{gather}
	\cD_{2n}:=\left\{ C=(c_1,\dots,c_{2n}) \mid c_i=c_{n+i}\text{ for $1\leq i\leq n$}, \wt(C)=0\xmod 4\right\}.
\end{gather}
We will abuse notation somewhat by also using $[i]$ to denote the following length $2n$ codewords,
\begin{gather}\label{eqn:[i]codeword}
[0]:= (0^{2n}), \quad
[1]:= (1^{n}0^n), \quad
[2]:= (0^{n-1}10^{n-1}1), \quad
[3]:= (1^{n-1}0^n 1).
\end{gather}

The next result may be checked directly, and smooths out any conflict between (\ref{eqn:[i]vector}) and (\ref{eqn:[i]codeword}).
\begin{lem}\label{lem:[i]to[i]}
With the above conventions, the image of $(D_{2n}+[i])/A_1^{2n}$ in $\F_2^{2n}$ is $\cD_{2n}+[i]$ for $i\in \{0,1,2,3\}$.
\end{lem}

The above discussion shows, in particular, that $A_1^{12}\cong \sqrt{2}\ZZ^{12}$ embeds in $D_{12}^+$. In \S\ref{sec:egs:gepK3} we will make use of the fact that $\sqrt{3}\ZZ^{12}$ also embeds in $D_{12}^+$. To see this recall that the {\em (extended) ternary Golay code} is a linear sub space $\cG<\F_3^{12}$ of dimension $6$ such that if 
\begin{gather}
C\cdot D:=\sum_i c_id_i
\end{gather} 
for $C=(c_1,\ldots,c_{12})$ and $D=(d_1,\ldots,d_{12})$ then $C\cdot D=0$ when $C,D\in \cG$, and no non-zero codeword $C\in \cG$ has less than six non-zero entries. These properties determine $\cG$ uniquely, up to permutations of coordinates, and multiplications of coordinates by $\pm 1$ (cf. e.g. \cite{CS}).

We will denote the elements of $\F_3$ by $\{0,+,-\}$ when convenient. To obtain an embedding of $\sqrt{3}\Z^{12}$ in $D_{12}^+$ fix a copy $\cG$ of the ternary Golay code in $\F_3^{12}$. 
Multiplying some components by $-1$ if necessary we may assume that $(+^{12})\in\cG$. Then there are exactly $11$ code words $C^i=(c^i_1,\ldots,c^i_{12})\in\cG$ such that the first entry of $C^i$ is $+1$, five further entries are $+1$, and the remaining six entries are $-1$. 
Set $C^{12}=(+^{12})$ and define $\lambda^i:=(\lambda^i_1,\ldots,\lambda^i_{12})$ for $1\leq i\leq 12$ by setting $\lambda^i_{j}=\pm\frac12$ when $c^i_j=\pm 1$. Then the $\lambda^i$ all belong to $D_{12}^+$, and satisfy $\lambda^i\cdot \lambda^j=3\delta_{ij}$. So the $\lambda^i$ constructed in this way generate a sub lattice of $D_{12}^+$ isomorphic to $\sqrt{3}\Z^{12}$.

The discriminant group of $\sqrt{3}\Z^{12}$ is $\F_3^{12}$, so it is natural to consider the image of $D_{12}^+/\sqrt{3}\Z^{12}$ in $\F_3^{12}$. Denote it $\cG_{12}^+$. Since $D_{12}^+$ is a self-dual lattice, $\cG_{12}^+$ is a linear subspace such that $C\cdot D=0$ for $C,D\in \cG_{12}^+$. From the fact that $\lambda\in D_{12}^+$ can only satisfy $\lambda\cdot\lambda\leq 1$ if $\lambda=0$ we obtain that $\cG_{12}^+$ has no non-zero words with fewer than six non-zero entries. Applying the uniqueness of the ternary Golay code we obtain the following result.
\begin{lem}\label{lem:D12+Gol}
The image of $D_{12}^+/\sqrt{3}\Z^{12}$ in $\F_3^{12}$ is a copy of the ternary Golay code.
\end{lem}

\subsection{Modularity}\label{sec:back:mod}

We now review some results on modularity for vertex operator superalgebras. 

Zhu proved \cite{Z} that certain trace functions on irreducible modules for suitable \voas{} span representations of the modular group $\SL_2(\ZZ)$. More general modularity results that incorporate twisted modules have been obtained by Dong--Li--Mason \cite{DLM2000}, Dong--Zhao \cite{DZ,DZ2}, and Van Ekeren \cite{Van}. We will use the extension of \cite{Z,DLM2000} to \vosas{} established in \cite{DZ}.
 
To describe the relevant results 
let $W=\bigoplus_{n\in\frac12\ZZ} W_n$ be a rational $C_2$-cofinite \vosa{} of CFT type. 
Let $I_0$ be an index set for the isomorphism classes of irreducible $W$-modules, let $I_1$ be an index set for the isomorphism classes of irreducible canonically twisted $W$-modules, and set $I:=I_0\cup I_1$. Write $M_i$ for a representative (untwisted or canonically twisted) $W$-module corresponding to $i\in I$, 
and choose a compatible superspace structure $M_i=M_i^\even\oplus M_i^\odd$ for each $i\in I$.
For $M$ an untwisted or canonically twisted $W$-module define {\em vertex operators on the torus} $Y[a,z]:M\to M\ldp z\rdp$ for $a\in W$ by requiring that $Y[a,z]=Y(a,e^z-1)e^{n z}$ when $a\in W_n$, and define $a_{[n]}\in{\operatorname{End}}(M)$ by requiring $Y[a,z]=\sum_k a_{[n]}z^{-n-1}$. Also, write $W=\bigoplus_{n\in \frac12\ZZ} W_{[n]}$ for the eigenspace decomposition of $W$ with respect to the action of $\tilde{\omega}_{[1]}$, where $\tilde{\omega}:=\omega-\frac{c}{24}{\bf v}$, and $\omega$ and ${\bf v}$ are the Virasoro and vacuum elements of $W$, respectively. 
\begin{thm}[\cite{DZ}]\label{thm:DZmod}
Suppose that $W$ is a rational $C_2$-cofinite \vosa{} with central charge $c$. 
Then with $I$ and $\{M_i\}_{i\in I}$ as above there are 
maps 
$\rho_{ij}:\SL_2(\ZZ)\to \CC$ for $i,j\in I$ such that if $v\in W_{[n]}$ for some $n\in \frac12\ZZ$ 
then
\begin{gather}\label{eqn:modtrW}
\tr_{M_i}\left.\left(o(v)(-1)^{\ell F} q^{L_0-\frac c{24}}\right)\right|_n
\left(\begin{smallmatrix}a&b\\c&d\end{smallmatrix}\right)
=\sum_{j\in I_{\tilde k}} 
\rho_{ij}\left(\begin{smallmatrix}a&b\\c&d\end{smallmatrix}\right)
\tr_{M_j} \left(o(v)(-1)^{\tilde\ell F}  q^{L_0-\frac c{24}}\right)
\end{gather}
for $\ell\in \ZZ/2\ZZ$ and $\left(\begin{smallmatrix}a&b\\c&d\end{smallmatrix}\right)\in \SL_2(\ZZ)$, where 
$\tilde k=1+a(k+1)+c(\ell+1)\xmod 2$ and $\tilde\ell=1+b(k+1)+d(\ell+1)\xmod 2$ when $i\in I_k$.
\end{thm}
In (\ref{eqn:modtrW}) we utilize the usual slash notation from modular forms, setting
\begin{gather}
	\left(f|_n\left(\begin{smallmatrix}a&b\\c&d\end{smallmatrix}\right)\right)(\tau):=f\left(\tfrac{a\tau+b}{c\tau+d}\right)\tfrac1{(c\tau +d)^n}
\end{gather}
for a holomorphic function $f:\mathbb{H}\to \CC$ and $\left(\begin{smallmatrix}a&b\\c&d\end{smallmatrix}\right)\in\SL_2(\ZZ)$. 
If $W$ is a rational $C_2$-cofinite \voa{} then we can apply Theorem \ref{thm:DZmod} to $W$ by regarding it as a \vosa{} with trivial odd part. Then $I_0=I_1$ and we recover a specialization of Zhu's results \cite{Z} for $W$ by taking $k=\ell=1$ in (\ref{eqn:modtrW}).

Resume the assumption that $W$ is a rational $C_2$-cofinite \vosa{} of CFT type. Define the {\em characters} of an untwisted or canonically twisted $W$-module $M$  by setting
\begin{gather}\label{eqn:chpmMmod}
	\ch^\pm[M](\tau):=\tr_M\left((\pm1)^Fq^{L_0-\frac{c}{24}}\right).
\end{gather}
Then taking $v$ to be the vacuum in Theorem \ref{thm:DZmod} we obtain 
\begin{gather}
	\ch^\epsilon[M_i]\left(\frac{a\tau+b}{c\tau+d}\right) = \sum_{j\in I_{\tilde k}} \rho_{ij}(\gamma)\ch^{\tilde \epsilon}[M_j](\tau),
\end{gather}
for $\left(\begin{smallmatrix}a&b\\c&d\end{smallmatrix}\right)\in \SL_2(\ZZ)$, where $\epsilon=(-1)^\ell$ and $\tilde\epsilon=(-1)^{\tilde\ell}$.

In this work we will be especially interested in the situation in which a 
superspace $M$ is a module for a tensor product $V'\otimes V''$ of \vosas{}, and is thus equipped with two commuting actions of the Virasoro algebra. We write $L'_n$ and $L''_n$ for the Virasoro operators corresponding to $V'$ and $V''$, respectively, and write $c'$ and $c''$ for the corresponding central charges. Then it is natural to consider the {\em refined characters}
\begin{gather}\label{eqn:chhM}
\chh^\pm[M](\tau',\tau'') := \tr_{M}\left( (\pm 1)^F{q'}^{L'_0-\frac{c'}{24}}{q''}^{L''_0-\frac{c''}{24}}\right)
\end{gather}
where $q':=e^{2\pi i \tau'}$ and $q'':=e^{2\pi i \tau''}$. Krauel--Miyamoto \cite{KMi} have shown how Zhu's theory \cite{Z} extends so as to yield a modularity result for such refined characters in the \voa{} case. By replacing Zhu's results with the appropriate \vosa{} counterparts from \cite{DZ} in the proof of Theorem 1 in \cite{KMi} we readily obtain a direct analogue for \vosas{}. Here we require the following special case of this.
\begin{thm}[\cite{DZ,KMi}]\label{thm:KM}
Let $W$ be a rational $C_2$-cofinite \voa{} and suppose that $\omega=\omega'+\omega''$ where $\omega$ is the conformal vector of $W$, and $\omega'$ and $\omega''$ generate commuting representations of the Virasoro algebra. 
Then with $I$, $\{M_i\}_{i\in I}$ and $\rho_{ij}(\gamma)$ as in (\ref{eqn:modtrW}) 
we have
\begin{gather}
\chh^\epsilon[M_i]\left(\frac{a\tau'+b}{c\tau'+d},\frac{a\tau''+b}{c\tau''+d}\right) = 
\sum_{j\in I_{\tilde k}} \rho_{ij}(\gamma)\chh^{\tilde \epsilon}[M_j](\tau',\tau'')
\end{gather}
for $\left(\begin{smallmatrix}a&b\\c&d\end{smallmatrix}\right)\in \SL_2(\ZZ)$. Here $\tilde k$ and $\tilde \ell$ are as in (\ref{eqn:modtrW}), for $\epsilon=(-1)^\ell$ and $\tilde\epsilon=(-1)^{\tilde\ell}$.
\end{thm}

\subsection{Superconformal Algebras}\label{sec:back:sca}

In this section we recall some properties of the small $\cN=4$ superconformal algebra from \cite{A,ET87,ET88,ET}. It is strongly generated by a Virasoro field $T$ of dimension $2$, four odd fields $G^a$ ($a=0,1,2,3$) of dimension $\frac32$, and three even fields $J^i$ ($i=1, 2, 3$) of dimension $1$.
Define
\begin{gather}
\alpha^i_{a, b} := \frac{1}{2} \left( \delta_{a, i}\delta_{b, 0} -\delta_{b, i}\delta_{a, 0}\right) +\frac{1}{2}\epsilon_{iab}
\end{gather}
where $\epsilon_{ijk}$ is totally antisymmetric for $i, j, k\in\{1, 2, 3\}$, normalized so that $\epsilon_{123}=1$, and defined to be zero if one of the indices is zero. Also let $k$ be a positive integer and set $c=6k$.
The operator product algebra in a representation with central charge $c$ is then
\begin{gather}
\begin{split}\label{eqn:N=4OPEs}
T(z)T(w) &\sim \frac{\frac 12 c}{(z-w)^4}+\frac{2T(w)}{(z-w)^2}+\frac{\partial_w T(w)}{(z-w)}, \\
T(z)G^a(w) &\sim \frac{\frac 32 G^a(w)}{(z-w)^2}+\frac{\partial_w G^a(w)}{(z-w)}, \\
T(z)J^i(w) &\sim \frac{J^i(w)}{(z-w)^2}+\frac{\partial_wJ^i(w)}{(z-w)}, \\
G^a(z)G^b(w) &\sim \frac{\frac 23 c \delta_{a, b}}{(z-w)^3}-\frac{\sum\limits_{i=1}^3 8\alpha^i_{a, b} J^i(w)}{(z-w)^2}+\frac{2\delta_{a, b} T(w)-\sum\limits_{i=1}^3 4\alpha^i_{a, b} \partial_wJ^i(w)}{(z-w)}, \\
J^i(z)G^a(w) &\sim \frac{\sum\limits_{b=0}^3 4\alpha^i_{a, b} G^b(w)}{(z-w)}, \\
J^i(z)J^j(w) &\sim \frac{-\frac 12 k \delta_{i, j}}{(z-w)^2}+\frac{\sum\limits_{k=1}^3 \epsilon_{ijk} J^k(w)}{(z-w)}.
\end{split}
\end{gather}
Equivalently, we have the following commutation relations of modes, where $m, n \in \mathbb Z$ and $r, s \in \mathbb Z+\frac{1}{2}$.
\begin{gather}\label{eq:commutationrelations} 
\begin{split}
[L_m, L_n]&= (m-n) L_{m+n}=\frac{c}{12} \left(M^3-m\right) \delta_{m+n, 0},\\
[L_m, G^a_r]&= \left(\frac{m}{2}-r\right) G^a_{m+r},\\
[L_m, J^i_n]&= -nJ^i_{m+n}, \\
\{ G^a_r, G^b_s\} &= 2\delta_{a, b}L_{r+s}-\sum_{i=1}^3 4(r-s) \alpha^i_{a, b}J^i_{r+s}+\frac{c}{3}\left(r^2-\frac{1}{4}\right) \delta_{a, b}\delta_{r+s, 0},\\
[J^i_m, G^a_r]&= \sum_{b=0}^3 \alpha^i_{a, b}G^b_{m+r},\\
[J^i_m, J^j_n]&= \sum_{k=1}^3 \epsilon_{ijk} J^k_{m+n} -m\frac{k}{2} \delta_{i, j}\delta_{m+n, 0}.
\end{split}
\end{gather}
Note that the currents $J^i(z)$ represent the affine Lie algebra of ${\mathfrak{sl}}_2$ at level $k$, and $k$ is a positive integer in any unitary representation.
Often one uses another basis for these; namely
\begin{gather}\label{eqn:N=4altJs}
J(z):= -2i J^1(z), \quad J^+(z):= J^2(z)-iJ^3(z), \quad J^-(z):= -J^2(z)-iJ^3(z).
\end{gather}
Also for the odd currents there is another useful basis:
\begin{gather}\label{eqn:N=4altGs}
G^{-,1}(z):= G^0(z)-iG^1(z), \quad 
G^{+,1}(z):= -G^2(z)+iG^3(z), \\
G^{-,2}(z):= G^2(z)+iG^3(z), \quad 
G^{+,2}(z):= -G^0(z)+iG^1(z),
\end{gather}
which yields
\begin{gather}\label{eqn:JGs}
[J^\pm_m, G^{\mp, x}_r] = G^{\pm, x}_{m+r}, \qquad [J_m, G^{\pm, x}_r]= \pm G^{\pm,  x}_{m+r}
\end{gather}
for $x\in\{1, 2\}$. 

The mode algebra of the affine Lie algebra of $\mathfrak{sl}_2$ at level $k$ has a family of automorphisms, called spectral flow. They are induced from affine Weyl translations so they are parameterized by the translation lattice which is isomorphic to $\mathbb Z$. For $\ell\in \ZZ$ the corresponding action is denoted $\sigma^\ell$, and satisfies
\begin{gather}\label{eqn:actspcflw}
\begin{split}
\sigma^\ell\left(J^\pm_n\right) &= J^\pm_{n\mp \ell}, \\
\sigma^\ell\left(J_n\right) &= J_{n} -\delta_{n, 0}\ell k , \\
\sigma^\ell\left(L_n\right) &= L_{n} +\delta_{n, 0}\left(\frac{\ell}{2} J_0 +\frac{\ell^2k}{4}\right). 
\end{split}
\end{gather}
We may consider twisted modules as in (2.10) of \cite{CR}. That is, given a module $M$ for the affine Lie algebra of $\mathfrak{sl}_2$ at level $k$ let $\sigma_\ell^*$ be the unique invertible linear transformation of $M$ such that the action of the mode algebra satisfies
\begin{gather}\label{eqn:spcflwtws}
X \sigma_\ell^*(v)= \sigma_\ell^*\left(\sigma^{-\ell}(X)v\right).
\end{gather} 
This $\sigma^\ell$-twisting of $M$ is isomorphic to $M$ as a module for the mode algebra, but since the grading is changed they are in general not isomorphic as modules for the vertex operator algebra $L_k(\mathfrak{sl}_2)$ attached to $\mathfrak{sl}_2$ at level $k$. Identification of the twisted modules can often be achieved via characters.

So we define a {character} of $M$ by setting
\begin{gather}\label{eqn:ch-full}
\text{ch}[M](y, u, \tau) := \tr_M\left(y^kz^{J_0}q^{L_0-\frac{c}{24}} \right)
\end{gather}
where $z=e^{2\pi i u}$ and $q=e^{2\pi i \tau}$.
Then from the action of $\sigma^\ell$ we see that 
\begin{gather}
\text{ch}[\sigma_\ell^*\left(M\right)](y, z, q) = \text{ch}[M](yz^\ell q^{\frac14 \ell^2}, zq^{\frac12 \ell}, q).
\end{gather}
If $k$ is a positive integer and $M$ is an irreducible integrable highest-weight module of level $k$, then the character of $M$ is a component of a vector-valued Jacobi form and the twisted module can be seen to be isomorphic to the original one if $\ell$ is even. If $\ell$ is odd it 
maps the irreducible integrable highest-weight module with highest weight $j$ (corresponding to the $j+1$ dimensional representation of $\mathfrak{sl}_2$) to the one with highest weight $k-j$. Thus it corresponds to fusion with the order two simple current module. 

We now consider a \vosa{} $V$ of central charge $c=6k$ that contains a commuting pair of sub \voas{} $L_k(\mathfrak{sl}_2)$ and $U$, where the Virasoro elements of $V$ and $L_k(\mathfrak{sl}_2)\otimes U$ coincide. We require $U$ to be rational, and suppose that $V$ is of the form
\begin{gather}
V=\bigoplus_{i} L_{k}(\lambda_i)\otimes U_i
\end{gather}
where $L_{k}(\lambda_i)$ is the irreducible $L_k(\mathfrak{sl}_2)$-module with highest weight $\lambda_i$, and $U_i$ is an irreducible $U$-module. Further we require that $V$ contain strong generators for the small $\cN=4$ superconformal algebra at $c=6k$. In this situation the parity of elements in $V$ is given by $(-1)^{J_0}$ where $J_0$ is as in (\ref{eqn:N=4altJs}), (\ref{eqn:JGs}).
The spectral flow automorphism $\sigma:=\sigma^1$ extends naturally to $V$ in such a way that
\begin{gather}
\sigma^*(V)=\bigoplus_{i} \left(L_{k}(\lambda_i)\boxtimes_{L_k(\mathfrak{sl}_2)} L_k(\lambda_*)\right)\otimes U_i
\end{gather}
as a $L_k(\mathfrak{sl}_2)\otimes U$-module, where $L_k(\lambda_*)$ denotes the order two simple current of $L_k(\mathfrak{sl}_2)$, and $\boxtimes_V$ is the fusion product for modules over $V$.

Under spectral flow we have $\sigma^{-1}(L_0)=L_0+\frac{1}{2}J_0+\frac{c}{24}$,
so the conformal weight of the modes $G^{\pm, x}_r$ becomes $r\pm\frac{1}{2}$ as endomorphisms on twisted modules. 
Especially, the $G^{\pm, x}_{\mp \frac{1}{2}}$ act with conformal dimension zero, and thus induce grading-preserving maps from the even part to the odd part of the twisted module.  
We compute
\begin{gather}
\begin{split}
\left(G^{+, 1}_{-\frac{1}{2}}+G^{-, 2}_{\frac{1}{2}}\right)^2
&=\frac{1}{2}\left\{G^{+, 1}_{-\frac{1}{2}}+G^{-, 2}_{\frac{1}{2}}, G^{+, 1}_{-\frac{1}{2}}+G^{-, 2}_{\frac{1}{2}}\right\}\\
&= -2L_0+J_0\\
&= 2\sigma^{-1}\left( L_0-\frac{c}{24}\right).
\end{split}
\end{gather}
This shows that $G^{+, 1}_{-\frac{1}{2}}+G^{-, 2}_{\frac{1}{2}}$ is an invertible map from the even conformal grade $n$ subspace of $\sigma^*(M)$ to the odd one as long as $n-\frac{c}{24}$ does not vanish. We summarize this as follows.
\begin{lem}\label{lem:G0}
If $M$ is a $V$-module and $\sigma^*(M)=\bigoplus_n \sigma^*(M)_n$ is the grading of the corresponding twisted module 
then $\operatorname{sdim}\sigma^*(M)_n=0$ unless $n-\frac{c}{24}=0$.
\end{lem}

\subsection{Free Fields}\label{sec:back:ff}

We now describe a free field realization of the small $\cN=4$ superconformal algebra at $c=6$. 
Consider the tensor product of the \vosa{} of four free fermions with the rank four Heisenberg \voa. An action of the compact Lie group $SU(2)$ is given by organizing both the fermions and bosons in the standard and conjugate representations of $SU(2)$. By Lemma 3.4 of \cite{CH} the sub \vosa{} of fixed points for this $SU(2)$ action contains the small $\cN=4$ superconformal algebra at $c=6$. So by using bosonization we obtain a free field realization of the small $\cN=4$ superconformal algebra at $c=6$ in terms of twelve free fermions. The precise expression for the latter can be found in \S2 of \cite{GTVW}. In terms of lattice \vosas{} this free field realization may be described as follows. 

The \vosa{} of four free fermions may be identified with the lattice \vosa{} associated to $\ZZ^2$. We have $\ZZ^2=D_2\cup \left(D_2+[2]\right)$ in the notation of \S\ref{sec:back:vosas}, and the exceptional isomorphism $D_2\cong A_1\oplus A_1$. So $V_{D_2}\cong V_{A_1}\otimes V_{A_1}$.
Four free fermions are thus the simple current extension of $V_{A_1}\otimes V_{A_1}$ by the unique simple current with dimension $\frac12$. The free field realization in terms of twelve free fermions can be inspected to be a sub \vosa{} of 
\begin{gather}
V_{D_2}^{\otimes 3} \oplus V_{D_2+[2]}^{\otimes 3},
\end{gather}
and we have just seen that this \vosa{} is isomorphic to 
\begin{gather}
V_{A_1}^{\otimes 6}\oplus V_{A_1+(1)}^{\otimes 6} 
\end{gather}
where $A_1+(1)$ denotes the unique non-trivial coset of $A_1$ in its dual.
For completeness we describe a precise realization. Consider odd fields $b_1, \dots, b_6$ and $c_1, \dots, c_6$ with operator products
\begin{gather}
b_i(z)c_j(w) \sim \frac{\delta_{i, j}}{(z-w)}, \qquad b_i(z)b_j(w) \sim c_i(z)c_j(w) \sim 0, 
\end{gather}
generating a copy $F(12)\cong V_{\ZZ^6}$ of the \vosa{} of $12$ free fermions. From the above we have $(V_{D_2}\oplus V_{D_2+[2]})^{\otimes 3}\cong V_{\ZZ^6}$.
Then $b_1, b_2, c_1, c_2$ span the first copy of $V_{D_2} \oplus V_{D_2+[2]}$, the $b_3, b_4, c_3, c_4$ span the second copy, and $b_5, b_6, c_5, c_6$ span the last one. 
The three fields
\begin{gather}
h=\,:\!b_1c_1\!:\,+\,:\!b_2c_2\!:\,, \qquad e=\,:\!b_1b_2\!:\,, \qquad f=\,:\!c_1c_2\!:\,,
\end{gather}
are all in the even sub \voa{} $V_{D_2}$ of the first copy of $V_{D_2} \oplus V_{D_2+[2]}$. These three fields strongly generate a \voa{} isomorphic to $L_1(\mathfrak{sl}_2)$.
The four fields
\begin{gather}
G^{+,1}= \,:\!b_1b_3b_5\!:\,, \quad G^{-,1}= \,:\!c_2b_3b_5\!:\,, \quad  G^{-,2}= \,:\!c_1c_3c_5\!:\,, \quad G^{+,2}= \,:\!b_2c_3c_5\!:\,,
\end{gather}
are all fields in $V_{D_2+[2]}^{\otimes 3}$. These seven fields together with the Virasoro field strongly generate a \vosa{} isomorphic to the small $\cN=4$ superconformal algebra at $c=6$.

\section{Self-Dual Vertex Operator Superalgebras}\label{sec:selfdual}

In this section we prove our first main result, which is a classification of self-dual $C_2$-cofinite vertex operator superalgebras of CFT type with central charge less than or equal to $12$. 

\begin{thm}\label{thm:selfdual-class}
If $W$ is a self-dual $C_2$-cofinite \vosa{} of CFT type with central charge $c\leq 12$ then either $W\cong F(n)$ for some $0\leq n\leq 24$, or $W\cong V_{E_8}\otimes F(n)$ for $0\leq n\leq 8$, or $W\cong V_{D_{12}^+}$.
\end{thm}

\begin{proof}
We first prove the claimed result for the special case that $c=12$. 
For this we require to show that a self-dual $C_2$-cofinite \vosa{} of CFT type with $c=12$ is isomorphic to 
one of $V_{D_{12}^+}$, $V_{E_8}\otimes F(8)$ or $F(24)$.

So let $W$ be a self-dual $C_2$-cofinite \vosa{} of CFT type with central charge $12$. We will constrain the possibilities for $W$ by extending the methods used in \S5.1 of \cite{Dun_VACo}. There, $V_{D_{12}^+}$ is characterized as the unique such \vosa{} with an $\cN=1$ superconformal structure and vanishing weight $\frac12$ subspace. We will also employ the arguments of \S4.2 of \cite{DM15}, in which the hypothesis of superconformal structure is removed.

We begin by applying Theorem \ref{thm:DZmod} to $W$. In this situation $I_0$ and $I_1$ are singletons. Let us set $I_k=\{k\}$, so that 
$M_0=W$ and $M_1=W_\tw$. Then taking $v$ to be the vacuum in (\ref{eqn:modtrW}) we obtain
that $Z_\NSNS^+(\tau):=\tr_Wq^{L_0-\frac12}$ is a weakly holomorphic modular form with character $\rho_{00}$ of weight $0$ for $\Gamma_\theta:=\lab S,T^2\rab$, where $S=\left(\begin{smallmatrix}0&-1\\1&0\end{smallmatrix}\right)$ and $T=\left(\begin{smallmatrix}1&1\\0&1\end{smallmatrix}\right)$. Further, $\rho_{00}$ is, a priori, trivial on $T^2$, and $\pm 1$ on $S$. This fact is the basis of the proof of Proposition 5.7 in \cite{Dun_VACo}, which shows that if $d:=\dim W_{\frac12}$ then 
\begin{gather}
\begin{split}\label{eqn:Z}
Z^+_\NSNS(\tau)&=\frac{\eta(\tau)^{24}}{\eta(2\tau)^{24}}+24+d\\
	&=q^{-\frac12}+d+276q^{\frac12}+O(q)
\end{split}
\end{gather}
so that in particular, $\rho_{00}$ 
is trivial on $\Gamma_\theta$.
Also, similar to Proposition 5.9 of \cite{Dun_VACo}, we find from taking $\gamma=TS$ in (\ref{eqn:modtrW}) that $(W_\tw)_n$ vanishes unless $n\geq \frac12$. For another 
application let $u,u'\in W_1$. Then taking $v=u_{[-1]}u'$ in (\ref{eqn:modtrW}), setting $X(\tau):=\tr_Wo(u_{[-1]}u')q^{L_0-\frac12}$ and using the triviality of $\rho_{00}$ on $\Gamma_\theta$ we find that $X(\tau)$ is a weakly holomorphic modular form of weight $2$ for $\Gamma_\theta$ satisfying $X(\tau)=Cq^{-\frac 12}+O(1)$ as $\tau\to i\infty$ for some $C\in \CC$. Note that $\Gamma_\theta$ has two cusps, represented by $i\infty$ and $1$. From (\ref{eqn:modtrW}) and the fact that $W_\tw$ has $L_0$-grading bounded below by $\frac12$ we see that $X(\tau)=O(1)$ as $\tau\to 1$. Since the space of modular forms of weight $2$ for $\Gamma_\theta$ is one-dimensional, spanned by $\theta_{D_4}(\frac{\tau+1}{2})=1-24q^{\frac12}+O(q)$ (cf. Theorem 7.1.6 in \cite{Rankin}), it follows that 
\begin{gather}
\begin{split}\label{eqn:X}
X(\tau)&=-2Cq\frac{{\rm d}}{{\rm d}q}Z_\NSNS^+(\tau)+D\theta_{D_4}(\tfrac{\tau+1}{2})\\
&=Cq^{-\frac12}+D+(-276C-24D)q^{\frac12}+O(q)
\end{split}
\end{gather} 
for some $C,D\in \CC$ (with $C$ as above), which we can expect will depend on $u$ and $u'$.

We now endeavour to connect $C$ and $D$ in (\ref{eqn:X}) to the Lie algebra structure on $W_1$. For this note that the first paragraph of the proof of Theorem 4.5 in \cite{DM15} applies to $W$, showing that $W$ admits a unique (up to scale) non-degenerate invariant bilinear form, which we henceforth denote $\lab\cdot\,,\cdot\rab$. 
We also have that $W_1$ is contained in the kernel of $L_1$, so by Theorem 1.1 of \cite{MR2097833} the Lie algebra structure on $W_1$ is reductive. 
Applying the argument of the proof of Theorem 5.12 in \cite{Dun_VACo} to $W$---this uses the identity (\ref{eqn:Z})---we see that the Lie rank of $W_1$ is bounded above by $12$.
We can identify a simple component of $W_1$ just by considering $d=\dim W_{\frac12}$. To do this note that the invariant bilinear form $\lab\cdot\,,\cdot\rab$ on $W$ is non-degenerate when restricted to $W_{\frac12}$. If we let $U$ denote the sub \vosa{} of $W$ generated by $W_{\frac12}$ then, arguing as in \cite{GS} (cf. also \S3 of \cite{Tam}), we see that $U$ is isomorphic to the Clifford algebra \vosa{} defined by the orthogonal space structure on $W_{\frac12}$, and if $V$ is the commutant of $U$ in $W$ then $W$ is naturally isomorphic to $U\otimes V$.
So $W_1$ contains $U_1$, which is the Lie algebra naturally associated to the orthogonal structure on $W_{\frac12}$. Since $W_1$ has Lie rank bounded above by $12$ we must have $d\leq 24$. Indeed, the case that $d=24$ is realized by $W=U=F(24)$. 

\begin{table}[!htbp]
    \caption{Dual Coxeter numbers of simple complex Lie algebras\label{tab:mfgh}}
    \medskip
\makebox[\linewidth]{
  \begin{tabular}{| c | c |}
    \cline{1-2}
     \multicolumn{1}{|c|}{$\mathfrak{g}$} & \multicolumn{1}{c|}{${h^\vee}$} \\ 
        \hline
    \toprule
    \hline
    $A_n$ & $n+1$ \\
    $B_n$ & $2n-1$ \\
    $C_n$ & $n+1$ \\
    $D_n$ & $2n-2$ \\
    $E_6$ & 12 \\
    $E_7$ & 18 \\
    $E_8$ & 30 \\
    $F_4$ & 9 \\
    $G_2$ & 4 \\
    \hline
    \end{tabular}
}
\end{table}

So suppose henceforth that $d<24$. Then $\dim U_1<\dim W_1=276$, and $V_1\neq \{0\}$. To further constrain $W_1$ we consider (\ref{eqn:X}) with $u,u'\in V_1$. 
Note that since $V_{\frac12}=\{0\}$ by construction we have $\tr_Wo(u)o(u')q^{L_0-\frac12}=\kappa(u,u')q^{\frac12}+O(q)$ where $\kappa$ is the Killing form on $W_1$. We have $u_{[1]}u'=\lab u,u'\rab\vv$ 
according to Lemma 5.1 of \cite{Dun_VACo}, so by an application of Sublemma 6.9 of \cite{DZ} we find that
\begin{gather}\label{eqn:AppZhuProp}
	\begin{split}
	X(\tau)&=
    \tr_Wo(u)o(u')q^{L_0-\frac12}-\frac1{12}{\langle u,u'\rangle}E_{2}(\tau)Z_\NSNS^+(\tau)\\
    &=-\frac1{12}\lab u,u'\rab q^{-\frac12}-\frac1{12}\lab u,u' \rab d+(\kappa(u,u')-21\lab u,u'\rab )q^{\frac12}+O(q)
    \end{split}
\end{gather}
where $E_2(\tau)=1-24\sum_{n>0}\frac{nq^n}{1-q^n}$ is the quasi-modular Eisenstein series of weight $2$, and $Z_\NSNS^+(\tau)$ is as in (\ref{eqn:Z}). Comparing (\ref{eqn:X}) with (\ref{eqn:AppZhuProp}) we find that $C=-\frac1{12}\lab u,u'\rab$, $D=-\frac1{12}\lab u,u'\rab d$, and 
\begin{gather}\label{eqn:kappad}
	\kappa(u,u')=(44+2d)\lab u,u'\rab.
\end{gather} 
In particular, the Killing form is non-degenerate on $V_1$, so the Lie algebra structure on $V_1$ is semisimple.
So let $\mfg$ be a simple component of $V_1$. From the main theorem of \cite{MR2219226} we know that the vertex operators attached to $V_1$ represent the affine Lie algebra associated to $\mfg$ with integral level; call it $k$. Then for $\alpha$ a long root of $\mfg$ we have $\kappa(\alpha,\alpha)=4h$ where $h$ is the dual Coxeter number of $\mfg$, and $\lab \alpha,\alpha\rab =2k$. So (\ref{eqn:kappad}) with $u=u'=\alpha$ yields $h=(22+d)k$. We are reduced to finding the pairs $(d,\mfg)$ where $d$ is an integer $0\leq d<24$, and $\mfg$ is a simple Lie algebra with Lie rank bounded above by $12-\frac12 d$ such that the dual Coxeter number $h^\vee$ of $\mfg$ is an integer multiple of $22+d$. 
Inspecting Table \ref{tab:mfgh} we see that 
either $d=0$ and $\mfg$ is of type $D_{12}$, or $d=8$ and $\mfg$ is of type $E_8$. The first of these is realized by $W=V_{D_{12}^+}$. The second is realized by $W=V_{E_8}\otimes F(8)$. Thus we have dealt with the special case that $c=12$.

To complete the proof let $W$ be a self-dual $C_2$-cofinite \vosa{} of CFT type with $c<12$. By applying Theorem 11.3 of \cite{DLM2000} to 
$W^\even$ we may conclude that $c$ is a rational number. Then the argument of Satz 2.2.2 of \cite{Hoehn2007} applies to $W$, and shows that $c\in \frac12\ZZ$. So $n=24-2c$ is a positive integer and $W':=W\otimes F(n)$ is a self-dual $C_2$-cofinite \vosa{} of CFT type with central charge $c'=12$. Since $n$ is positive $W'$ is one of $V_{E_8}\otimes F(8)$ or $F(24)$, by what we have already proved about self-dual \vosas{} with central charge $12$. The desired conclusion follows. 
\end{proof}

\section{Superconformal Field Theory}\label{sec:bscft}

\subsection{Potential Bulk Conformal Field Theory}\label{sec:bcft:pbcft}

The basic structure underlying a {\em bulk \cft{}} is a module $\mathcal H$ for a tensor product $V'\otimes V''$ of \voas{} $V'$ and $V''$. It is required that 
\begin{gather}\label{eqn:HVV}
\mathcal H = \bigoplus_{i} N'_i \otimes N''_i
\end{gather}
where the $N'_i$ and $N''_i$ are irreducible 
modules for $V'$ and $V''$, respectively. 
Also, $V'\otimes V''$ should appear exactly once as a summand of $\mathcal H$. 
A standard but very special case is that $V'= V''$ is rational and $C_2$-cofinite, and the $N'_i=N''_i$ are all the irreducible $V'$-modules. We call this the {\em diagonal} \cft{} for $V=V'=V''$. 

Various further properties are required of $\mathcal H$, including {\em closure under fusion} and {\em modular invariance}. Define 
\begin{gather}
\begin{split}
\chh[\cH](\tau',\tau''):&=\tr_\cH\left( {q'}^{L'_0-\frac{c'}{24}}{q''}^{L''_0-\frac{c''}{24}}\right)\\
&=\sum_i \ch[N_i'](\tau')\ch[N_i''](\tau'')
\end{split}
\end{gather}
(cf. (\ref{eqn:chhM})).
Modular invariance is the requirement that the {\em partition function}
\begin{gather}\label{eqn:ZH}
\begin{split}
Z_{\mathcal H}(\tau) :&= 
\chh[\cH](\tau,-\bar\tau)\\
&=\sum_{i} \ch[N'_i](\tau) \ch[N''_i](-\bar\tau) 
\end{split}
\end{gather}
be invariant for the natural action of $\SL_2(\ZZ)$. 
Closure under fusion is the requirement that operator products close on $\mathcal H$. So superficially at least, fusion and modularity force a bulk \cft{} to resemble a self-dual \voa. 
This is the motivation for the Main Question we formulated in \S\ref{sec:intro}.

We will answer the Main Question positively in what follows, using a certain convenient substitute for the notion of bulk \cft{}. Also, 
we will find that there are more examples if we allow the \voa{} in the question to be a vertex operator superalgebra.

To motivate 
our approach 
note that modular invariance is a 
strong constraint 
that is used in practice to classify possible examples of \cfts{}. There are additional properties of correlation functions that are harder to verify (cf. \cite{TW17,W-snap}), but from a representation category point of view 
these correlation requirements are satisfied by symmetric special Frobenius algebra objects in the modular tensor category of a suitably chosen \voa{} according to \cite{FRS}. In this work lattice \voas{} underly all examples, so all simple objects are simple currents, and if a symmetric special Frobenius algebra object $\mathcal A$ is a direct sum of inequivalent simple currents then there is a rather explicit prescription for the decomposition (\ref{eqn:HVV}); namely (5.85) of \cite{FRS}. For example, if we assume $V'\cong V''$ and choose $\mathcal A=V'$ 
then 
the partition function of the bulk is the charge conjugation invariant. In the cases we consider every irreducible $V'$-module 
will be invariant under charge conjugation, so the charge conjugation invariant will coincide with the ordinary diagonal 
modular invariant.
Motivated by this we introduce the following.
\begin{defn}\label{def:pbcft}
A {\em potential bulk \cft{}} is a $V'\otimes V''$-module $\mathcal H$ as in (\ref{eqn:HVV}) such that the partition function (\ref{eqn:ZH}) is modular invariant. 
\end{defn}

Now to formulate an answer to the Main Question, suppose that $W$ is a self-dual $C_2$-cofinite \vosa{} such that 
\begin{gather}\label{eqn:VAB}
W\cong \bigoplus_i N'_i\otimes N''_i
\end{gather}
as a $V'\otimes V''$-module, for $V'$ and $V''$ a commuting pair of rational $C_2$-cofinite sub \voas{}, where the $N'_i$ and $N''_i$ are 
irreducible modules for $V'$ and $V''$, respectively. Define 
 \begin{gather}\label{eqn:ZW}
 Z_{W}(\tau) := 
\chh^+[W](\tau,-\bar\tau)
\end{gather}
where $\chh^\pm[\,\cdot\,]$ is as in (\ref{eqn:chhM}).
\begin{prop}\label{prop:pbcft2cond}
With $W$ as in (\ref{eqn:VAB}), if the $S$-matrix of $V''$ is real 
and the eigenvalues of the action of 
$L_{0}'-L_{0}''$ on $W$ belong to $\ZZ+\frac1{24}(c'-c'')$
then $Z_W$ is modular invariant.
\end{prop}
\begin{proof}
Since $W$ is self-dual Theorem \ref{thm:DZmod} implies that $\ch^+[W]$ (cf. (\ref{eqn:chpmMmod})) is invariant under $S$. It follows from Theorem \ref{thm:KM} then that $\chh^+[W](\tau', \tau'')$ is $S$-invariant as well. In particular we have
\begin{gather}
\begin{split}\label{eqn:chhWSS}
\chh^+[W]\left(-\frac{1}{\tau'}, -\frac{1}{\tau''}\right) 
&=\sum_{i, j, \ell} S'_{i j} \ch^+[N_j'](\tau') S''_{i \ell} \ch^+[N_\ell''](\tau'')\\
&= \sum_{i}  \ch^+[N_i'](\tau')  \ch^+[N_i''](\tau'')
\end{split}
\end{gather}
for suitable matrices $S'$ and $S''$. 
 Consider the action of $S$ on $Z_W(\tau)$. 
Using 
(\ref{eqn:chhWSS}) and the hypothesis that the entries of $S''$ are real we compute
 \begin{gather}
 \begin{split}
Z_W\left(-\frac1\tau\right)
 &= \sum_i \ch^+[N_i']\left(-\frac{1}{\tau}\right)  \overline{\ch^+[N_i'']\left(-\frac{1}{\tau}\right)}\\
&= \sum_{i, j, \ell} S'_{i, j} \ch^+[N_j'](\tau) \overline{S''_{i, \ell} \ch^+[N_\ell''](\tau)}\\
&= \sum_{i, j, \ell} S'_{i, j} \ch^+[N_j'](\tau) S''_{i, \ell} \ch^+[N_\ell''](-\bar\tau)\\
&= \sum_{i}  \ch^+[N_i'](\tau)  \ch^+[N_i''](-\bar\tau)\\
&=
Z_W(\tau).
  \end{split}
  \end{gather} 
So $Z_W$ is $S$-invariant. Invariance under $T$ follows from
\begin{gather}
\begin{split}
Z_W(\tau+1)&=
\tr_W\left( {e}^{2\pi i (\tau+1)(L'_0-\frac{c'}{24})}{e}^{-2\pi i (\bar \tau +1)( L''_0-\frac{c''}{24})}\right)
\\
&= e^{\frac{2\pi i}{24} (c''-c')}\tr_W\left(e^{2\pi i (L_0'-L_0'')} {q}^{L'_0-\frac{c'}{24}}{\bar q}^{L''_0-\frac{c''}{24}}\right)
\end{split}
\end{gather}
and our hypothessis on the eigenvalues of $L_0'$ and $L_0''$.
This completes the proof.
\end{proof}

Proposition \ref{prop:pbcft2cond} gives a path to answering the Main Question positively, in the sense that if its hypotheses are satisfied then it identifies a self-dual \vosa{} $W$ with a potential bulk \cft{} $\mathcal H$ as a module for the underlying \voa{} $V'\otimes V''$. The two conditions of Proposition \ref{prop:pbcft2cond} are strong, but are satisfied in interesting cases as we will see in \S\ref{sec:egs}.

\subsection{Potential Bulk Superconformal Field Theory}\label{sec:bscft:pbscft}

We are also interested in relating self-dual \vosas{} to \scfts{} in this work. To define a supersymmetric counterpart to the notion of potential bulk \cft{} we consider a $V'\otimes V''$-module 
\begin{gather}\label{eqn:HVVsup}
\mathcal H = \bigoplus_{i} N'_i \otimes N''_i
\end{gather}
as in (\ref{eqn:HVV}), but allow $V'$ and $V''$ to be \vosas{}, and allow the $N'_i$ in (\ref{eqn:HVV}) to be irreducible untwisted or canonically twisted modules for $V'$, and similarly for the $N''_i$. (Strictly speaking, $\cH$ is a module for the even sub \voa{} of $V'\otimes V''$.) Following tradition we use subscripts $\text{NS}$ and $\text{R}$ to indicate restrictions to untwisted and canonically twisted modules for $V'$ and $V''$. So, 
\begin{gather}
\mathcal H_\NSNS := \bigoplus_{\substack{i\\ \text{$N'_i$ untwisted}\\ \text{$N''_i$ untwisted} }} N'_i \otimes N''_i ,\\
\mathcal H_{\text{NS-R}} := \bigoplus_{\substack{i\\ \text{$N'_i$ untwisted}\\ \text{$N''_i$ twisted} }} N'_i \otimes N''_i ,
\end{gather}
and so on. We call $\cH_\NSNS$ the $\NSNS$ sector, \&c.
We also assume that $\cH$ is equipped with a compatible superspace structure $\cH=\cH^\even\oplus \cH^\odd$, so that $\cH$ is graded by $(\Z/2\Z)^3$, 
\begin{gather}
\mathcal H = 
\mathcal H^\text{even}_\NSNS \oplus \mathcal H^\text{odd}_\NSNS\oplus \mathcal H^\text{even}_{\text{NS-R}} \oplus \dots \oplus \cH^\text{even}_{\text{R-R}}\oplus\cH^\odd_\RR.
\end{gather}
We may regard a bulk \cft{} as a bulk \scft{} in which only the even part of the $\NSNS$ sector is non-zero.
From this point of view it is natural to expect examples in which the $\NSNS$ sector of a \scft{} is identified with a self-dual \vosa{}, and the $\RR$ sector is identified with its canonically twisted module. 
As such, the prescription (\ref{eqn:ZH}) does not usually define a modular invariant function when the superspace structure is non-trivial. Rather, Theorem \ref{thm:DZmod} indicates that we should consider the vector-valued function
\begin{gather}\label{eqn:ZcHsup}
	Z_\cH(\tau):=
	\begin{pmatrix} 
	Z^+_\NSNS(\tau) \\
	Z^-_\NSNS(\tau) \\
	Z^+_\RR(\tau)  \\
	Z^-_\RR(\tau)
	\end{pmatrix}
\end{gather}
where $Z^\pm_\text{X-Y}$ is defined by setting
	$Z^\pm_\text{X-Y}(\tau):=
	\chh^\pm[\cH_\text{X-Y}](\tau,-\bar\tau)$
(cf. (\ref{eqn:chhM})) for $\text{X}, \text{Y}\in\{\text{NS},\text{R}\}$. 
Then modularity for $Z_\cH$ is the requirement that 
\begin{gather}\label{eq:ZcHmodsup}
\begin{split}
{\bf S}\cdot
{Z}_\cH\left(-\frac{1}{\tau}\right)
&=
{Z}_\cH\left(\tau\right),
\\
{\bf T}\cdot
{Z}_\cH\left(\tau+1\right)&=
{Z}_\cH\left(\tau\right),
\end{split}
\end{gather}
where 
\begin{gather}\label{eqn:STsup}
{\bf S}:=
\begin{pmatrix} 
	1 & 0 & 0 & 0 \\
	0 & 0 & 1 & 0 \\
	0 & 1 & 0 & 0 \\
	0 & 0 & 0 & 1 
	\end{pmatrix},\quad
{\bf T}:=
\begin{pmatrix} 
	0 & 1 & 0 & 0 \\
	1 & 0 & 0 & 0 \\
	0 & 0 & 1 & 0 \\
	0 & 0 & 0 & 1 
	\end{pmatrix}.
\end{gather}
Motivated by this we formulate the following super-analogue of Definition \ref{def:pbcft}.
\begin{defn}\label{def:pbscft}
A {\em potential bulk \scft{}} is a $V'\otimes V''$-module $\mathcal H$ as in (\ref{eqn:HVVsup}) whose partition function (\ref{eqn:ZcHsup}) satisfies (\ref{eq:ZcHmodsup}).
\end{defn}

\begin{rem}
For bulk \scfts{} 
we expect that correlation function requirements are encoded by a suitably formulated notion of symmetric special Frobenius {\em superalgebra} object. 
It would be interesting to generalize \cite{FRS} to the super setting, and determine whether the examples we describe in this work are compatible or not. Superalgebra objects are introduced in the context of vertex algebra tensor categories in \cite{CKL, CKM}.
\end{rem}

To formulate a supersymmetric counterpart to Proposition \ref{prop:pbcft2cond} consider a self-dual $C_2$-cofinite \vosa{} $W$ such that 
\begin{gather}\label{eqn:WNNsup}
W\cong \bigoplus_i N'_i\otimes N''_i
\end{gather}
as a $V'\otimes V''$-module, for $V'$ and $V''$ a commuting pair of rational $C_2$-cofinite sub \vosas{}, where the $N'_i$ and $N''_i$ are irreducible (untwisted) modules for $V'$ and $V''$, respectively. Write $W_\tw$ for the unique irreducible canonically twisted $W$-module (cf. \S\ref{sec:back:vosas}) and chose a superspace structure $W_\tw=W_\tw^\even\oplus W_\tw^\odd$ that is compatible with the superspace structure on $W$. 
We then define $Z_W(\tau)$ in analogy with (\ref{eqn:ZcHsup}), so that
\begin{gather}\label{eqn:ZWsup}
	Z_W(\tau):=
	\begin{pmatrix} 
	Z^+(\tau) \\
	Z^-(\tau) \\
	Z^+_\tw(\tau)  \\
	Z^-_\tw(\tau)
	\end{pmatrix}
\end{gather}
where 
	$Z^\pm(\tau):=
	\chh^\pm[W](\tau,-\bar\tau)$ and
	$Z^\pm_\tw(\tau):=
	\chh^\pm[W_\tw^\even](\tau,-\bar\tau)$.
The proof of the next result follows in a directly similar way to that of Proposition \ref{prop:pbcft2cond}.
\begin{prop}\label{prop:pbscftconds}
With $W$ as in (\ref{eqn:WNNsup}), if the $S$-matrix of $V''$ is real, 
the eigenvalues of $L'_0-L''_0$ on $W^\even$ lie in $\ZZ+\frac1{24}(c'-c'')$, and the eigenvalues of $L'_0-L''_0$ on $W^\odd$ lie in $\ZZ+\frac12+\frac1{24}(c'-c'')$ then $Z_W$ satisfies the modularity condition (\ref{eq:ZcHmodsup}).
\end{prop}

Thus a \vosa{} $W$ satisfying the hypotheses of Proposition \ref{prop:pbscftconds} also answers the Main Question positively, in the sense that it is identified with the $\NSNS$ sector of a potential bulk \scft{} as a module for the underlying \vosa{} $V'\otimes V''$, and similarly for $W_\tw$ and the $\RR$ sector. 
Note that reality of the modular $S$-matrix holds for $L_k(\mathfrak{sl}_2)$ when $k$ is positive and integral, and holds also for minimal models of the Virasoro algebra. Inspecting Weil's description of the modular group action on theta functions for cosets of an even lattice $L$ in its dual $L^*$, we see that if the inner products on $L^*$ are contained in $\frac12\Z$ then the $S$-matrix for the lattice \voa{} $V_L$ is also real. This holds in particular for $L=D_{4n}$, which will play a prominent role in what follows.

\subsection{Superconformal structure}\label{sec:bscft:supstruc}

Superconformal field theories are usually assumed to come equipped with supersymmetry. 
With $\cH$ as in (\ref{eqn:HVVsup}) we define an {\em $\cN=(2,2)$ superconformal structure} 
to be an identification of sub \vosas{} of $V'$ and $V''$ with vacuum modules for the $\cN=2$ superconformal algebra. 
In this situation it is natural to consider
\begin{gather}\label{eqn:ZcHusup}
	Z_\cH(u,\tau):=
	\begin{pmatrix} 
	Z^+_\NSNS(u,\tau) \\
	Z^-_\NSNS(u,\tau) \\
	Z^+_\RR(u,\tau)  \\
	Z^-_\RR(u,\tau)
	\end{pmatrix}
\end{gather}
where 
	$Z^\pm_\text{X-Y}(u,\tau):=
	\chh^\pm[\cH_\text{X-Y}](u,\tau,-\bar u,-\bar\tau)$
and
\begin{gather}\label{eqn:chhMsup}
	\chh^\pm[M](u',\tau',u'',\tau''):=
	\tr_M\left((\pm 1)^F {z'}^{J'_0}  {q'}^{L'_0-\frac{c'}{24}} {z''}^{J''_0}{q''}^{L''_0-\frac{ c''}{24}} \right).
\end{gather}
Here $z'=e^{2\pi i u'}$ and $q'=e^{2\pi i \tau'}$, \&c. 
We also have the {\em elliptic genus} of $\cH$, defined by setting
\begin{gather}
\EG_\cH(u,\tau):=
	\chh^-[\cH_\RR](u,\tau,0,-\bar\tau).
\end{gather} 

Now the modularity conditions on $\cH$ are richer. For simplicity we describe them in the special 
case that $c'=c''$. Then the natural counterpart to (\ref{eq:ZcHmodsup}) is 
\begin{gather}\label{eq:bulkmodular}
\begin{split}
e^{-\pi i \frac c6 \left(\frac{u^2}{\tau}-\frac{\bar u^2}{\bar\tau}\right)}
{\bf S}\cdot
{Z}_\cH\left(\frac{u}\tau,-\frac{1}{\tau}\right)
&=
{Z}_\cH\left(u,\tau\right),
\\
{\bf T}\cdot
{Z}_\cH\left(u,\tau+1\right)&=
{Z}_\cH\left(u,\tau\right),
\end{split}
\end{gather}
where $c=c'=c''$, and ${\bf S}$ and ${\bf T}$ are as in (\ref{eqn:STsup}). Call this modularity for $Z_\cH$. If $\cH$ is a \scft{} underlying a sigma model with Calabi--Yau target space $X$ then $\EG_\cH$ is also modular, in the sense that we have
\begin{gather}
\begin{split}\label{eqn:EGcHmod}
e^{-\pi i  d \frac{u^2}{\tau}}
\EG_\cH\left(\frac{u}\tau,-\frac{1}{\tau}\right)
&=
\EG_\cH\left(u,\tau\right),
\\
\EG_\cH\left(u,\tau+1\right)&=
\EG_\cH\left(u,\tau\right),
\end{split}
\end{gather}
where $d=\dim_\CC X$, and $\tau\mapsto \EG_\cH(u,\tau)$ remains bounded as $\tau\to i\infty$, for any fixed $u\in\CC$. That is, $\EG_\cH$ is a {\em weak Jacobi form} of weight $0$ and index $\frac12\dim_\CC X$. In this situation $\EG_\cH(0,\tau)$ is the Euler characteristic of $X$. (In particular, $\EG_\cH(0,\tau)$ is a constant function of $\tau$.) Let us say that $\cH$ is {\em elliptic} if $\EG_\cH$ satisfies (\ref{eqn:EGcHmod}).  
Say that $\cH$ satisfies {\em spectral flow symmetry} if 
\begin{gather}\label{eqn:specflowsym}
Z^\pm_{\RR}(u', \tau', u'',\tau'') =  (z'z'')^{\frac {c}{6}}(q' q'')^{\frac {c}{24}} Z^\pm_{\NSNS}\left(u'+\tfrac12{\tau'}, \tau', u''+\tfrac12{\tau''},\tau''\right)
\end{gather}
where $c=c'=c''$. 

At this point it is natural to define a {\em potential bulk $\cN=(2,2)$ \scft{}} to be a potential bulk \scft{} $\cH$ (cf. Definition \ref{def:pbscft}) with $\cN=(2,2)$ superconformal structure such that (\ref{eq:bulkmodular}), (\ref{eqn:EGcHmod}) and (\ref{eqn:specflowsym}) are satisfied. These are strict requirements. Interestingly, we will see in examples that the extra superconformal structure allows us to weaken these requirements in what seems to be a useful way. In anticipation of this we offer the following.
\begin{defn}\label{def:qpb22scft}
Say that $\cH$ as in (\ref{eqn:HVVsup}) is a {\em quasi potential bulk $\cN=(2,2)$ \scft{}} if $\cH$ is elliptic (\ref{eqn:EGcHmod}), satisfies spectral flow symmetry (\ref{eqn:specflowsym}), and if $Z_\cH$ is modular (\ref{eq:bulkmodular}) for some finite index subgroup of the modular group.
\end{defn}
So for the notion of quasi potential bulk $\cN=(2,2)$ \scft{} we relax condition (\ref{eq:bulkmodular}), which is the invariance of $Z_\cH$ under the action of $\SL_2(\ZZ)$ defined by the left hand sides of (\ref{eq:bulkmodular}), and require invariance only for some subgroup $\Gamma<\SL_2(\ZZ)$ such that the coset space $\Gamma\backslash\SL_2(\ZZ)$ is finite.

We may also be interested in the case that
both $V'$ and $V''$ contain the small $\cN=4$ superconformal algebra at some central charge $c$. In this situation we call $\cH$ a {\em potential bulk $\cN=(4,4)$ \scft{}} if (\ref{eq:bulkmodular}), (\ref{eqn:EGcHmod}) and (\ref{eqn:specflowsym}) hold, and if in addition $c=6k$ for $k$ a positive integer, and spectral flow for each of $V'$ and $V''$ is realized by fusion with the order two simple current of $L_k(\mathfrak{sl}_2)$ (as discussed in \S\ref{sec:back:sca}). For the notion of {\em quasi potential bulk $\cN=(4,4)$ \scft{}} we relax the requirement of modularity (\ref{eq:bulkmodular}) of $Z_\cH$ from the modular group to some finite index subgroup.

In order to present a counterpart to Proposition \ref{prop:pbscftconds} we now consider a self-dual $C_2$-cofinite \vosa{} $W$ with a decomposition as in (\ref{eqn:WNNsup}) such that 
$V'$ and $V''$ contain copies of the vacuum module for the $\cN=2$ superconformal algebra at some central charge $c=c'=c''$.
We then define $Z_W(u,\tau)$ in analogy with (\ref{eqn:ZcHusup}), setting
\begin{gather}\label{eqn:ZWsupN2}
	Z_W(u,\tau):=
	\begin{pmatrix} 
	Z^+(u,\tau) \\
	Z^-(u,\tau) \\
	Z^+_\tw(u,\tau)  \\
	Z^-_\tw(u,\tau)
	\end{pmatrix}
\end{gather}
where 
	$Z^\pm(u,\tau):=
	\chh^\pm[W](u,\tau,-\bar u,-\bar\tau)$ and
	$Z^\pm_\tw(u,\tau):=
	\chh^\pm[W_\tw^\even](u,\tau,-\bar u,-\bar\tau)$,
and $\chh^\pm[\,\cdot\,]$ is as in (\ref{eqn:chhMsup}).
The proof of the next result is similar to the proofs of Propositions \ref{prop:pbcft2cond} and \ref{prop:pbscftconds}, but note the restriction that $c'=c''$.
\begin{prop}\label{prop:pb22scftconds}
Suppose that $W$ is as in (\ref{eqn:WNNsup}), and $V'$ and $V''$ contain the vacuum module for the $\cN=2$ superconformal algebra at some central charge $c=c'=c''$. 
If $L'_0-L''_0$ acts on $W^\even$ and $W_\tw$ with eigenvalues in $\ZZ$, and on $W^\odd$ with eigenvalues in $\ZZ+\frac12$, and if the $S$-matrices of $V'$ and $V''$ are both real, 
then $Z_W$ satisfies the modularity and spectral flow symmetry conditions, (\ref{eq:bulkmodular}) and (\ref{eqn:specflowsym}).
\end{prop}

Define the {\em elliptic genus} of $W$ by setting
\begin{gather}\label{eqn:EGW}
\EG_W(u,\tau):=
	\chh^-[W_\tw](u,\tau,0,-\bar\tau).
\end{gather}
We will presently see examples $W$ for which Proposition \ref{prop:pb22scftconds} fails but spectral flow symmetry holds and $\EG_W$ is a weak Jacobi form of weight $0$ and some index.

\section{Examples}\label{sec:egs}

\subsection{Type D Conformal Field Theory}\label{sec:egs:typeDcft}

Let $n$ be a positive integer. The bulk of the diagonal \cft{} wth $V'=V''=V_{D_{2n}}$ is
\begin{gather}
	\cH = \bigoplus_{i=0}^3 V_{D_{2n}+[i]}\otimes V_{D_{2n}+[i]}.
\end{gather}
Observe that we may embed $D_{2n}\oplus D_{2n}$ in $D_{4n}$ by taking the first copy of $D_{2n}$ to be the vectors in $D_{4n}$ supported on the first $2n$ components, and letting the second copy be its orthogonal complement. Then for the cosets of $D_{4n}$ in its dual we have
\begin{gather}\label{eqn:D4ntoD2nD2n}
\begin{split}
D_{4n}+[0]&=\left(D_{2n}+[0], D_{2n}+[0]\right) \cup \left(D_{2n}+[2], D_{2n}+[2]\right), \\
D_{4n}+[1]&=\left(D_{2n}+[1], D_{2n}+[1]\right) \cup \left(D_{2n}+[3], D_{2n}+[3]\right), \\
D_{4n}+[2]&=\left(D_{2n}+[0], D_{2n}+[2]\right) \cup \left(D_{2n}+[2], D_{2n}+[0]\right), \\
D_{4n}+[3]&=\left(D_{2n}+[1], D_{2n}+[3]\right) \cup \left(D_{2n}+[3], D_{2n}+[1]\right). \\
\end{split}
\end{gather}
From this we immediately obtain the following result.
\begin{prop}\label{prop:VD2ncft}
For $n$ a positive integer, the bulk of the diagonal $V_{D_{2n}}$ \cft{} is isomorphic to $V_{D_{4n}^+}$ as a $V_{D_{2n}}\otimes V_{D_{2n}}$-module.
\end{prop}
Note that ${D_4^+}\simeq \ZZ^4$ and ${D_8^+}\simeq {E_8}$. So Proposition \ref{prop:VD2ncft} furnishes bulk \cft{} interpretations for the self-dual \vosas{} $F(8)$, $V_{E_8}$ and $V_{D_{12}^+}$ appearing in Theorem \ref{thm:selfdual-class}.

An embedding of $A_1^{2n}$ in $D_{2n}$ is discussed in \S\ref{sec:back:vosas}. We may use this to formulate a counterpart to Proposition \ref{prop:VD2ncft} for tensor powers of $L_1(\mathfrak{sl}_2)$. For example, 
if $\cH$ denotes the bulk of the diagonal conformal field theory of $L_1(\mathfrak{sl}_2)^{\otimes 2n}$ then by virtue of the isomorphism $V_{A_1}\cong L_1(\mathfrak{sl}_2)$ we have
\begin{gather}
\begin{split}
	\cH&=\bigoplus_{C\in \F_2^{2n}} V_{A_1^{2n}+C}\otimes V_{A_1^{2n}+C}\\
	&\cong\bigoplus_{C\in \mathcal{Z}_{4n}}V_{A_1^{4n}+C}
\end{split}
\end{gather} 
where 
\begin{gather}
	\mathcal{Z}_{2m}:=\left\{ C=(c_1,\dots,c_{2m})\in\F_2^{2m} \mid c_i=c_{m+i}\text{ for }1\leq i\leq m\right\}.
\end{gather}
Observe that $\mathcal{Z}_{2m}=\cD_{2m}\cup\cD_{2m}+[2]$ in the notation of (\ref{eqn:[i]codeword}). By applying Lemma \ref{lem:[i]to[i]} and noting that $D_{2m}\cup D_{2m}+[2]\cong\ZZ^{2m}$ we obtain the following result.
\begin{prop}\label{prop:VA2ncft}
For $n$ a positive integer, the bulk of the diagonal $L_1(\mathfrak{sl}_2)^{\otimes 2n}$ \cft{} is isomorphic to $F(8n)$ as a $L_1(\mathfrak{sl}_2)^{\otimes 2n}\otimes L_1(\mathfrak{sl}_2)^{\otimes 2n}$-module.
\end{prop}
Proposition \ref{prop:VA2ncft} furnishes bulk \cft{} interpretations for the \vosas{} $F(8)$, $F(16)$ and $F(24)$ in Theorem \ref{thm:selfdual-class}.

It is instructive to consider the analogue of Proposition \ref{prop:VA2ncft} where $L_1(\mathfrak{sl}_2)^{\otimes 2n}$ is replaced by a simple current extension. Write $(1^{2n})$ as a shorthand for the ``all ones'' vector $(1,1,\ldots,1)\in \F_2^{2n}$. We will consider $V'\cong V''\cong V_L$ where $L=A_1^{2n}\cup A_1^{2n}+(1^{2n})$.
Observe that the irreducible $V_L$-modules are the 
$V_{L+C}=V_{A_1^{2n}+C}\oplus V_{A_1^{2n}+(1^{2n})+C}$ for $C\in \F_2^{2n}$ with $\wt(C)=0\xmod 2$. 
For simplicity assume that $n$ is even, so that $L$ is an even lattice and $V_L$ is a \voa{}.
Then for the bulk of the diagonal \cft{} for $V_L$ we have
\begin{gather}
\begin{split}
	\cH&=\bigoplus_{\substack{C\in \F_2^{2n}\\\wt(C)=0\xmod 2\\c_{2n}=0}} V_{L+C}\otimes V_{L+C}\\
	&= \bigoplus_{\substack{C\in \F_2^{2n}\\\wt(C)=0\xmod 2}}\left(V_{A_1^{2n}+C}\otimes V_{A_1^{2n}+C}
	\oplus V_{A_1^{2n}+(1^{2n})+C}\otimes V_{A_1^{2n}+C}\right).
\end{split}
\end{gather}
Comparing with (\ref{eqn:[i]codeword}) we see that
\begin{gather}
	\cH\cong \bigoplus_{C\in \cD_{4n}^+} V_{A_1^{4n}+C}
\end{gather}
where $\cD_{4n}^+:=\cD_{4n}\cup \cD_{4n}+[1]$. Since $D_{4n}^+=D_{4n}\cup D_{4n}+[1]$ by definition, an application of Lemma \ref{lem:[i]to[i]} yields the following alternative interpretation for $V_{D_{4n}^+}$ as a potential bulk \cft{}, at least for $n$ even.
\begin{prop}\label{prop:diagcftVL}
Let $n$ be an even positive integer and let $L=A_1^{2n}\cup A_1^{2n}+(1^{2n})$. Then the bulk diagonal \cft{} associated to $V_L$ is isomorphic to $V_{D_{4n}^+}$ as a $V_L\otimes V_L$-module.
\end{prop}
Proposition \ref{prop:diagcftVL} offers a bulk \cft{} interpretation for $V_{E_8}$ distinct from that of Proposition \ref{prop:VD2ncft}.

\subsection{Type D Superconformal Field Theory}\label{sec:egs:typeDscft}

We now consider the diagonal \scft{} associated to $2n$ free fermions. By the boson-fermion correspondence we have $F(2n)\cong V_{D_n}\oplus V_{D_n+[2]}$, 
and $F(2n)_\tw\cong V_{D_n+[1]}\oplus V_{D_n+[3]}$. 
So we have
\begin{gather}
	\cH_\NSNS=F(2n)\otimes F(2n),\quad
	\cH_\RR=F(2n)_\tw\otimes F(2n)_\tw,
\end{gather}
in this case.
This gives us a bulk \scft{} interpretation for $F(4n)$ for $n$ a positive integer (cf. Proposition \ref{prop:VA2ncft}).
\begin{prop}\label{prop:diagFn}
Let $n$ be a positive integer. Then the $\NSNS$ sector of the bulk diagonal \scft{} associated to $2n$ free fermions is isomorphic to $F(4n)$ as a $F(2n)\otimes F(2n)$-module, and the $\RR$ sector is isomorphic to $F(4n)_\tw$ as a canonically twisted $F(2n)\otimes F(2n)$-module. 
\end{prop}

Next we consider the diagonal \scft{} associated to the $D_{2n}$ torus. In this case $V'=V''=V_{D_{2n}}\otimes F(2n)$, so 
\begin{gather}
	\cH_\NSNS=\bigoplus_i \left(V_{D_{2n}+[i]}\otimes F(2n)\right)\otimes \left(V_{D_{2n}+[i]}\otimes F(2n)\right)
\end{gather}
as modules for $V'\otimes V''$, and 
\begin{gather}
	\cH_\RR=\bigoplus_i \left(V_{D_{2n}+[i]}\otimes F(2n)_\tw\right)\otimes \left(V_{D_{2n}+[i]}\otimes F(2n)_\tw\right)
\end{gather}
as canonically twisted modules for $V'\otimes V''$. Comparing this description with the decomposition of $D_{4n}^+=D_{4n}\cup D_{4n}+[1]$ into cosets for $D_{2n}\oplus D_{2n}$ (cf. (\ref{eqn:D4ntoD2nD2n})) we obtain the following identification.
\begin{prop}
Let $n$ be a positive integer and set $V=V_{D_{2n}}\otimes F(2n)$. Then the $\NSNS$ sector of the bulk diagonal \scft{} associated to the $D_{2n}$ torus is isomorphic to $V_{D_{4n}^+}\otimes F(4n)$ as a $V\otimes V$-module, and the $\RR$ sector is isomorphic to $V_{D_{4n}^+}\otimes F(4n)_\tw$ as a canonically twisted $V\otimes V$-module. 
\end{prop}
We also obtain our first example of a potential bulk \scft{} with superconformal structure.
\begin{thm}\label{thm:D4n+F4n}
Let $n$ be a positive integer. Then the \vosa{} $W=V_{D_{4n}^+}\otimes F(4n)$ is a potential bulk $\cN=(2,2)$ \scft{} in the sense of \S\ref{sec:bscft}, for $V'\cong V''\cong V_{D_{2n}}\otimes F(2n)$, and the elliptic genus defined by this structure vanishes.
\end{thm}
\begin{proof}
The modularity (\ref{eq:bulkmodular}) 
and spectral flow symmetry (\ref{eqn:specflowsym}) follow from Proposition \ref{prop:pb22scftconds}. We may compute directly that the elliptic genus (\ref{eqn:EGW}) vanishes, so it trivially satisfies (\ref{eqn:EGcHmod}). 
\end{proof}

Taking $n=2$ in Theorem \ref{thm:D4n+F4n} we obtain an interpretation for $V_{E_8}\otimes F(8)$ as the bulk \scft{} of the sigma model with $D_{4}$ torus as target. Volpato has observed \cite{V} that supersymmetry preserving symmetry groups of sigma models with arbitrary $4$-dimensional torus as target space embed in the Weyl group of $E_8$. The Weyl group of $E_8$ acts naturally on $V_{E_8}\otimes F(8)$. It appears that $V_{E_8}\otimes F(8)$ can play the analogous role for sigma models on $4$-dimensional tori that $V_{D_{12}^+}$ has been shown \cite{DM16} to play for sigma models on K3 surfaces.

\subsection{Type A Superconformal Field Theory}\label{sec:egs:typeak3}

In principle we may consider \scfts{} with $V'\cong V''\cong V_L$ for $L=A_1^{2n}\cup A_1^{2n}+(1^{2n})$ also for $n$ odd (cf. Proposition \ref{prop:diagcftVL}). However, as we have explained in \S\ref{sec:back:ff}, the case $n=3$ is distinguished by the presence of $\cN=4$ superconformal structure. This fact underpins a significant part of the analysis of \cite{GTVW}, in which it is shown that  the diagonal \scft{} with $V'\cong V''\cong V_L$ underlies a Kummer type K3 sigma model arising from the canonical $\Z/2\Z$-orbifold of the $D_4$ torus. This is the {\em tetrahedral} K3 sigma model in \cite{TW13}. 
We discuss this sigma model from the point of view of $A_1^6$ now. 
Although the motivation of \cite{GTVW} is somewhat different, their detailed analysis precedes, and may serve to flesh out the discussion we offer here. 

So let $L=A_1^6\cup A_1^6+(1^6)$ in this section. Directly applying the observations preceding Proposition \ref{prop:diagcftVL} with $n=3$ we see that the $\NSNS$ sector of the diagonal \scft{} associated to $V_L$ satisfies $\cH_\NSNS\cong V_{D_{12}}\oplus V_{D_{12}+[1]}=V_{D_{12}^+}$ as a $V_L\otimes V_L$-module. Noting that the irreducible canonically twisted $V_L$-modules are the $V_{L+C}$ with $C\in \F_2^6$ and $\wt(C)=1\xmod 2$ we find that $\cH_\RR\cong V_{D_{12}+[2]}\oplus V_{D_{12}+[3]}$ as a canonically twisted $V_L\otimes V_L$-module. The diagonal \scft{} for $V_L$ underlies the tetrahedal K3 sigma model according to \cite{GTVW}, so we have the following super-analogue of Proposition \ref{prop:diagcftVL} for $n=3$.
\begin{prop}\label{prop:egs:typeak3}
For $L=A_1^6\cup A_1^6+(1^6)$, the $\NSNS$ sector of the tetrahedral K3 sigma model is isomorphic to $V_{D_{12}^+}$ as a $V_L\otimes V_L$-module, and the $\RR$ sector of the tetrahedral K3 sigma model is isomorphic to $V_{D_{12}+[2]}\oplus V_{D_{12}+[3]}$ as a canonically-twisted $V_L\otimes V_L$-module.
\end{prop}

Now let us consider superconformal structure. As explained in \S\ref{sec:back:ff} the \vosa{} $V_L\cong V_{A_1^6}\oplus V_{A_1^6+(1^6)}$ contains a copy of the vacuum module for the small $\cN=4$ superconformal algebra at $c=6$. 
As in \S\ref{sec:back:ff} we choose the first copy of $L_1(\mathfrak{sl}_2)$ in $L_1(\mathfrak{sl}_2)^{\otimes 6}<V_L$ 
to generate the affine $\mathfrak{sl}_2$ sub algebra. Then spectral flow corresponds to fusion with $V_{L+C}\otimes V_{L+C}$ where $C=(10^5)$. In terms of $V_{D_{12}^+}$ this is the same as fusion with $V_{D_{12}+[2]}$ by force of (\ref{eqn:[i]codeword}). Thus spectral flow interchanges the $\NSNS$ and $\RR$ sectors. An explicit calculation, such as is carried out in \S D.3 of \cite{GTVW}, verifies that $\EG_W(u,\tau)$ is a weak Jacobi form of weight $0$ and index $1$ such that $\EG_W(0,\tau)=24$. Thus we have the following.
\begin{thm}\label{thm:D12+LL}
The \vosa{} $V_{D_{12}^+}$ is a potential bulk $\cN=(4,4)$ \scft{} in the sense of \S\ref{sec:bscft}, for $V'\cong V''\cong V_L$, and the elliptic genus defined by this structure is the K3 elliptic genus.
\end{thm}

Note that $V_{D_{12}^+}$ serves as the moonshine module for Conway's group \cite{Dun_VACo,DM15}, and is precisely the \vosa{} that is used to attach weak Jacobi forms with level to supersymmetry preserving symmetries of K3 sigma models in \cite{DM16}.
Results equivalent to Proposition \ref{prop:egs:typeak3} and Theorem \ref{thm:D12+LL} have been obtained independently in \cite{TW17}.

\subsection{Gepner Type Superconformal Field Theory}\label{sec:egs:gepK3}

We now present a \scft{} interpretation for $V_{D_{12}^+}$ of a different nature. 
As mentioned in \S\ref{sec:back:vosas}, the lattice \vosa{} $V_{\sqrt{3}\Z}$ realizes the vacuum module of the $\cN=2$ superconformal algebra at $c=1$. In this section we set $K=\sqrt{3}\Z^6$. Thus $V_K$ contains the vacuum module of the $\cN=2$ superconformal algebra at $c=6$. We will show that $W=V_{D_{12}^+}$ is a quasi potential bulk $\cN=(2,2)$ \scft{} (cf. Definition \ref{def:qpb22scft}) for $V'\cong V''\cong V_K$. It will develop that $W$ is closely related to the Gepner model of type $(1)^6$, which is a \scft{} that also has $V'\cong V''\cong V_K$.

Recall from \S\ref{sec:back:vosas} that the lattice $\sqrt{3}\Z^{12}\cong K\oplus K$ embeds in $D_{12}^+$. Using such an embedding we may fix commuting sub \vosas{} $V', V''<W$ such that $V'\cong V''\cong V_K$. Since the discriminant group of $K\oplus K$ is $\F_3^{12}$ it is natural to use ternary codewords of length $12$ to describe the irreducible $V_K\otimes V_K$-modules. According to Lemma \ref{lem:D12+Gol} the $V_K\otimes V_K$-modules that appear in $W$ are indexed by the codewords in a copy $\cG_{12}^+$ of the ternary Golay code. To make this more concrete let us assume that $(+^6-^6)$ is a word in $\cG_{12}^+$ (if not then permute the coordinates so as to make this true), and define $C',C''\in\F_3^6$ for $C\in \cG_{12}^+$ by setting 
$C'=(c_1,\ldots,c_6)$ and $C''=(c_7,\ldots,c_{12})$ when $C=(c_1,\ldots,c_{12})$. Then we have
\begin{gather}\label{eqn:D12+KK}
	W=\bigoplus_{C\in\cG_{12}^+} V_{K+C'}\otimes V_{K+C''}
\end{gather}
as modules for $V_K\otimes V_K$. Our main result in this section is the following.
\begin{thm}\label{thm:WKK}
The \vosa{} $V_{D_{12}^+}$ is a quasi potential bulk $\cN=(2,2)$ \scft{} in the sense of \S\ref{sec:bscft}, for $V'\cong V''\cong V_K$, and the elliptic genus defined by this structure is the K3 elliptic genus.
\end{thm}
\begin{proof}
Let $w$ be the marked complete weight enumerator of $\cG_{12}^+$ for the marking $C\mapsto (C',C'')$. That is, define $w$ to be the $6$-variate polynomial
\begin{gather}
	w(X',Y',Z',X'',Y'',Z''):=\sum_{C\in \cG_{12}^+} w_{C'}(X',Y',Z')w_{C''}(X'',Y'',Z'')
\end{gather}
where $w_C(X,Y,Z)$, for $C\in \mathbb{F}_3^n$, is defined by $w_C(X,Y,Z):=X^{a_0}Y^{a_+}Z^{a_-}$ in case $C$ has $a_0$ entries equal to $0$, and $a_{\pm}$ entries equal to $\pm1$. Then the functions $\chh^\pm[W](u',\tau',u'',\tau'')$ and $\chh^\pm[W_\tw](u',\tau',u'',\tau'')$ (cf. (\ref{eqn:chhMsup})) are obtained by substituting characters of suitable irreducible modules for the $\cN=2$ superconformal algebra at $c=1$. These characters can be expressed in terms of classical theta functions and the Dedekind eta function (cf. \cite{RY}) so they are invariant for the action of some finite index subgroup of $\SL_2(\ZZ)$. So $Z_W$ (cf. (\ref{eqn:ZW})) is invariant for some finite index subgroup of $\SL_2(\ZZ)$. Also, the $\cN=2$ characters satisfy spectral flow symmetry, so the same is true for $W$. 

It remains to examine the elliptic genus $\EG_W$ (cf. (\ref{eqn:EGW})) of $W$. For this define 
\begin{gather}\label{eqn:fs}
f_s(u,\tau):=\eta(\tau)^{-1}\sum_{k\in\ZZ} (e^{\pi i} z)^{k+\frac s6}q^{\frac32(k+\frac s6)^2},
\end{gather}
which is a character for the $\cN=2$ superconformal algebra at $c=1$ when $s\in \ZZ$. 
Note that $f_s$ depends only on $s\xmod 6$, we have $f_s(0,\tau)=e^{\pm \frac{\pi i}6}$ when $s=\pm 1\xmod 6$, and $f_s(0,\tau)$ vanishes identically when $s=3\xmod 6$. Because of this we have
\begin{gather}
	\EG_W(u,\tau)=w(f_3(u,\tau),f_1(u,\tau),f_{-1}(u,\tau),0,e^{\frac{\pi i}{6}},e^{-\frac{\pi i}{6}}).
\end{gather}
Under our hypotheses on $\cG_{12}^+$ the marked complete weight enumerator is given by
\begin{gather}
\begin{split}
	w(X',Y',Z',&X'',Y'',Z'')=({X'}^6+{Y'}^6+{Z'}^6)({X''}^6+{Y''}^6+{Z''}^6)\\
	&+90({X'}^4{Y'}{Z'}+{X'}{Y'}^4{Z'}+{X'}{Y'}{Z'}^4){X''}^2{Y''}^2{Z''}^2\\
	&+20({X'}^3{Y'}^3+{X'}^3{Z'}^3+{Y'}^3{Z'}^3)({X''}^3{Y''}^3+{X''}^3{Z''}^3+{Y''}^3{Z''}^3)\\
	&+90{X'}^2{Y'}^2{Z'}^2({X''}^4{Y''}{Z''}+{X''}{Y''}^4{Z''}+{X''}{Y''}{Z''}^4).
\end{split}
\end{gather}
So we have
\begin{gather}\label{eqn:EGWgolsmp}
	\EG_W=-2({f_3}^6+{f_1}^6+{f_{-1}}^6)
	+20({f_1}^3{f_{-1}}^3+{f_3}^3{f_1}^3+{f_{-1}}^3{f_{3}}^3).
\end{gather}
We may check directly using (\ref{eqn:fs}) that this expression for $\EG_W$ is a weak Jacobi form of weight $0$ and index $1$. Substituting $u=0$ in (\ref{eqn:EGWgolsmp}) we obtain $\EG_W(0,\tau)=-2((-1)+(-1))+20(1)=24$. So $\EG_W$ is indeed the K3 elliptic genus. This completes the proof.
\end{proof}

It is interesting to compare the quasi potential bulk \scft{} structure (\ref{eqn:D12+KK}) to the \scft{} underlying the Gepner model of type $(1)^6$, which realizes a K3 sigma model and also has $V'\cong V''\cong V_K$. To describe the bulk define a ternary code $\mathcal{U}$ of length $6$ by setting
\begin{gather}
\mathcal{U}:=\left\{C=(c_1,\ldots,c_6)\in \F_3^6\mid \sum_i c_i=0\right\}.
\end{gather} 
Then, according to \cite{GHV} for example, the $\NSNS$ sector is $\cH_\NSNS=\bigoplus_{a\in\F_3}\cH_\NSNS^a$ where
\begin{gather}\label{eqn:GepHNSNS}
\cH_\NSNS^a:=\bigoplus_{C\in \mathcal{U}} V_{K+C+(a^6)}\otimes V_{K+C-(a^6)}
\end{gather}
as a module for $V_K\otimes V_K$.
Since there are $V_K\otimes V_K$-modules in (\ref{eqn:GepHNSNS}) whose corresponding codewords in $\F_3^{12}$ have fewer than $6$ non-zero entries, the $V_K\otimes V_K$-module structure on $W$ in Theorem \ref{thm:WKK} does not identify it with the \scft{} underlying the Gepner model $(1)^6$. 

However, $W$ and the $(1)^6$ model have closely related symmetry. To see this let $G^{\cN=2}$ be the group of symmetries of the bulk \scft{} of the $(1)^6$ model that fix the states of the left and right moving $\cN=2$ superconformal algebras at $c=6$. According to \cite{GHV} this group is a split extension of $S_6$ by 
$(\ZZ/3\ZZ)^6$. The subgroup of the automorphism group of $\cG_{12}^+$ that stabilizes the splitting $C\mapsto (C',C'')$ (cf. (\ref{eqn:D12+KK})) is another copy of $S_6$. From this it follows that $G^{\cN=2}$ also acts on $W$, preserving the states of the two commuting $\cN=2$ superconformal algebras at $c=6$ in $V_K\otimes V_K$. 

Now let $G^{\cN=4}$ be the subgroup of $G^{\cN=2}$ that preserves the left and right moving copies of the small $\cN=4$ superconformal algebra. According to the results of \cite{DM16} the actions of $G^{\cN=4}$ on $W$ and the bulk of the $(1)^6$ model give the same equivariant elliptic genera. We hope that the notion of quasi potential bulk \scft{} will help to identify an abstract mechanism that explains this coincidence.

There are a number of $c=6$ combinations of minimal $\cN=2$ models $V$ such that $V\otimes V$ embeds in $V_{D_{12}^+}$. Does every such embedding give rise to a quasi potential $\cN=(2,2)$ \scft{} with a close relationship to the corresponding Gepner model? It would be valuable to have a complete classification of these embeddings.

\section*{Acknowledgements}

The authors are grateful to Anne Taormina and Katrin Wendland for correspondence regarding their recent manuscript \cite{TW17}.
T.C. appreciates correspondence with Geoff Mason, and gratefully acknowledges support from the Natural Sciences and Engineering Research Council of Canada (RES0020460). J.D. thanks Anne Taormina and Katrin Wendland for discussions on related topics, and gratefully acknowledges support from the U.S. National Science Foundation (DMS 1203162, DMS 1601306), and the Simons Foundation (\#316779).


\end{document}